\documentclass[onecolumn,ldraftcls, 11pt]{IEEEtran} 
\usepackage{xcolor}
\definecolor{plotblue}{HTML}{1F77B4}   
\definecolor{plotorange}{HTML}{FF7F0E} %
\usepackage{pgfplots}
\pgfplotsset{compat=newest}
\usepackage{float}
\usepackage[utf8]{inputenc} 
\usepackage[T1]{fontenc}
\usepackage{url}
\usepackage{ifthen}
\usepackage{cite}
\usepackage{algorithm}
\usepackage{algpseudocode}
\usepackage[cmex10]{amsmath}
\usepackage{wrapfig}
\usepackage{tikz}
\usepackage{circuitikz}
\usepackage{multirow}
\usepackage{makecell}
\usetikzlibrary{shapes, arrows, positioning, patterns}
\usepackage[includefoot, left=27mm,right=27mm,top = 25mm,  bottom=25mm]{geometry}

\usepackage{cite}
\usepackage[colorlinks=false]{hyperref}%
\usepackage{mleftright}       
\mleftright                   
\usepackage{subcaption}
\usepackage{graphicx}         
\usepackage{booktabs}         
\usepackage{amsthm}
 \usepackage{graphicx}
\usepackage{amsmath}
\usepackage{graphics,graphicx,psfrag,color,float}
\usepackage{epsfig,bbold,bbm}
\usepackage{bm}
\usepackage{mathtools}
\usepackage{amsmath}
\usepackage{amssymb}
\usepackage{amsfonts}
\usepackage{amsfonts}
\usepackage {times}
\usepackage{bbm}
\usepackage{soul}
\usepackage{cleveref}
\usepackage{enumitem}
\usepackage{xcolor}
\usepackage{txfonts}
\makeatletter
\def\namedlabel#1#2{\begingroup
    #2%
    \def\@currentlabel{#2}%
    \phantomsection\label{#1}\endgroup
}
\makeatother
\DeclareUnicodeCharacter{0325}{}

\def\QED{\mbox{\rule[0pt]{1.5ex}{1.5ex}}}

\newtheorem{theorem}{\bf{Theorem}}

\newcommand{\supp}{\operatorname{supp}}
\newcommand{\RNum}[1]{\uppercase\expandafter{\romannumeral #1\relax}}

\newcommand{\beq}{\begin{equation}}
\newcommand{\enq}{\end{equation}}
\newcommand{\bel}{\begin{lemma}}
\newcommand{\enl}{\end{lemma}}

\newcommand{\bet}{\begin{theorem}}
\newcommand{\ent}{\end{theorem}}

\newcommand{\tr}{\mathrm{Tr}}
\newcommand{\Te}{\bm{\mathfrak{Te}}}
\newcommand{\Tr}{\bm{\mathfrak{Tr}}}

\newcommand{\nn}{\nonumber}

\newcommand{\ketbra}[1]{|#1\rangle\langle#1|}

\newcommand{\floor}[1]{\left\lfloor #1 \right\rfloor}

\newcommand{\customfootnotetext}[2]{{
  \renewcommand{\thefootnote}{#1}
  \footnotetext[0]{#2}}}
\newcommand{\eps}{\varepsilon}

\newcommand*{\bbR}{\mathbb{R}}

\newcommand*{\bbI}{\mathbb{I}}

\newcommand*{\bbN}{\mathbb{N}}

\newcommand*{\bbE}{\mathbb{E}}

\newcommand*{\cA}{\mathcal{A}}

\newcommand*{\cH}{\mathcal{H}}

\newcommand*{\cD}{\mathcal{D}}

\newcommand*{\cN}{\mathcal{N}}

\newcommand*{\cS}{\mathcal{S}}

\newcommand*{\cT}{\mathcal{T}}

\newcommand*{\cW}{\mathcal{W}}
\newcommand*{\cZ}{\mathcal{Z}}
\newcommand*{\cE}{\mathcal{E}}

\newcommand{\ket}[1]{|#1 \rangle}

\mathchardef\mhyphen="2D

\newcommand{\ketbrasys}[2]{\overset{#2}{|#1\rangle\langle#1|}}

\newcommand{\abs}[1]{\left\vert#1\right\vert}

\makeatletter
\newcommand*{\rom}[1]{\expandafter\@slowromancap\romannumeral #1@}
\makeatother
\mathchardef\mhyphen="2D
\newlist{steps}{enumerate}{1}
\setlist[steps, 1]{leftmargin = 1.1cm, label = Step \arabic*.}
\newtheorem{remark}{Remark}
\newtheorem{definition}{Definition}
\newtheorem{claim}{Claim}

\newtheorem{example}[theorem]{Example}
\newtheorem{lemma}{Lemma}
\newtheorem{corollary}{Corollary}

\usepackage[utf8]{inputenc} 
\usepackage[T1]{fontenc}
\usepackage{url}
\usepackage{ifthen}
\usepackage[cmex10]{amsmath}
\usepackage{amsfonts} 
\usepackage{times}

\begin{document}

\title{Privacy Implies Stability: Information-Theoretic Generalization Bounds for Quantum Learning}

\author{
Ayanava Dasgupta\textsuperscript{$*$},  Naqueeb Ahmad Warsi\textsuperscript{$*$} and  Masahito Hayashi\textsuperscript{$\dagger$}

}

\customfootnotetext{$*$}{
Indian Statistical Institute,
Kolkata 700108, India.
Email: 
{\sf 
 ayanavadasgupta\_r@isical.ac.in, naqueebwarsi@isical.ac.in
}
}

\customfootnotetext{$\dagger$}{School of Data Science, The Chinese University of Hong Kong, Shenzhen, Longgang District, Shenzhen, 518172, China

International Quantum Academy, Futian District, Shenzhen 518048, China

Graduate School of Mathematics, Nagoya University, Nagoya, 464-8602, Japan.
Email: 
{\sf hmasahito@cuhk.edu.cn}}

\maketitle

\begin{abstract}
We develop an information-theoretic framework connecting stability, privacy, and generalization for quantum learning algorithms. Learning procedures are modeled as quantum instruments with classical-quantum outputs, and losses are represented by observables. We prove that under a classical-quantum sub-Gaussian condition, an information-theoretic stability measure controls the expected generalization error. Furthermore, we establish a high-probability generalization bound using quantum R\'enyi divergences to manage higher-order dependencies under non-commutativity.

In the trusted Data Processor setting, quantum differential privacy (QDP) provides a mechanism for stability. We show that one-neighbor QDP strictly bounds the information leaked by the classical-quantum output. Combining this with our stability theorem yields a direct privacy-to-generalization guarantee.

We also explore an untrusted Data Processor setting. Here, output privacy alone is insufficient since an adversarial processor could perform a highly informative procedure before applying noisy post-processing. To combat this, we introduce Information-Theoretic Admissibility (ITA), a certification condition ensuring the prescribed procedure is not just a degraded version of a strictly more informative, physically allowed operation on the encoded ensemble. We prove a fundamental separation: while admissibility and privacy are in strong tension in classical models, quantum non-orthogonality makes them compatible. A quantum measurement can be ITA - exhausting all relevant accessible information - without perfectly recovering the classical dataset. We illustrate this separation through a concrete quantum ITA example.
\end{abstract}


\section{Introduction}
\label{sec_intro}
Quantum learning algorithms operate on information that is intrinsically
physical.  The training data may be encoded in quantum states, the learning
procedure is naturally described by a quantum channel or, more generally, by
a quantum instrument, and the performance of the learned hypothesis is
evaluated through observables.  In such a setting, generalization is not only
a question about the statistical dependence between a classical dataset and
a classical hypothesis.  A learning procedure may also release, retain, or
correlate with a residual quantum system, and this system can carry
information about the training data even when the announced classical
hypothesis appears stable.

A basic feature distinguishing the quantum setting from the classical one is
that information is constrained by state distinguishability.  
Non-orthogonal quantum states cannot be perfectly discriminated, as quantified by the
Helstrom theory of quantum detection \cite{Helstrom1967_2} and by the theory
of accessible information
\cite{fuchs1996distinguishabilityaccessibleinformationquantum}.  Consequently, even an optimal measurement need not fully reveal the underlying classical data encoded into quantum systems.  
This makes it necessary to formulate
information leakage, stability, and privacy at the level of
classical--quantum states, rather than only at the level of classical
input-output distributions.

Quantum differential privacy provides one way of imposing such an
indistinguishability constraint.  It extends the classical requirement that
neighboring datasets produce nearly indistinguishable outputs to the setting
where the outputs are quantum states and the adversary may perform arbitrary
physically allowed measurements \cite{ZY2017,HRF23,AK25,Nuradha25}.
Recent work has developed several complementary formulations of quantum
privacy, including information-theoretic QDP, local quantum differential
privacy, private quantum channels, and privacy notions based on quantum
hypothesis testing and contraction of distinguishability
\cite{HRF23,AK25,Nuradha25,DHLYT22}.
From this viewpoint, privacy is not only a semantic protection guarantee;
it is also a quantitative restriction on how much statistical
distinguishability can survive a physically allowed quantum processing.

This distinguishability-based viewpoint is closely related to the ordering of
information in statistical experiments and channels.
In the classical theory, the comparison of experiments is formalized by Blackwell's order
\cite{Blackwell51,BG54}.  Quantum analogues are more subtle and have been
studied through comparison theorems, randomization criteria, quantum
statistical morphisms, and reverse data-processing formulations
\cite{Buscemi2012,Jencova2012,Matsumoto:2015vgh,Buscemi2016,Buscemi_2018,BST2019}.
These comparison and post-processing viewpoints give a precise way to
express when one procedure is more informative than another on the same
encoded ensemble.

This ordering viewpoint is particularly relevant to recent work on quantum
information ordering and differential privacy \cite{DWH2025}.  There,
privacy constraints are studied as restrictions on the informativeness of
quantum channels or statistical experiments, using distinguishability
measures such as hypothesis-testing divergences, hockey-stick divergences,
and related quantum $f$-divergences
\cite{Wang_2012,Sharma09,SW12,GLW24}.
The present paper uses this ordering perspective in a different direction:
we study the operational consequences of privacy and information constraints
for quantum learning. In particular, the ordering language is used not to
characterize an extremal or most-informative private mechanism, but to
formulate a certification condition for privacy claims in trusted and
untrusted Data Processor models.

This work is also inspired by the classical information-theoretic theory of
generalization.  In classical statistical learning, algorithmic stability
explains why empirical performance can predict performance on fresh data: if
the output of a learning algorithm is insensitive to small perturbations of
the training set, then empirical and population losses are close
\cite{Bousquet02}.  Mutual-information refinements quantify such stability
by measuring the dependence between the training data and the algorithmic
output \cite{XR_2017}.  Differential privacy, originally introduced as a
rigorous privacy guarantee for statistical databases \cite{dwork2006},
provides an important mechanism for enforcing this kind of stability and
thereby controlling generalization
\cite{Dwork2015Generalization,Bassily2016,RGBS21}.  The goal of this paper
is to develop the corresponding stability--privacy--generalization framework
for quantum learning algorithms.
A complementary line of work studies locally differentially private
mechanisms through extremal mechanisms, strong data-processing inequalities,
and contraction coefficients \cite{KOV15,Asoodeh24,ZA23}.
These works quantify how a privacy-constrained mechanism reduces
distinguishability or information at the output. This viewpoint is related
to the present paper because our privacy-to-stability result gives a direct
upper bound on the classical--quantum information that the released output
can retain about the training data.

These observations show that privacy in quantum learning is not merely a
property of an abstract channel in isolation.  It depends on which systems
are released, which party is trusted, and which measurements or
post-processings are physically available to the adversary.  We therefore
organize the operational meaning of privacy using three roles.
The
\emph{Respondent} provides the data.  The \emph{Data Processor} executes
the learning procedure.
The \emph{Investigator} receives the released
output and uses it for inference, prediction, or testing.
This distinction
is useful because privacy is not merely a property of an abstract channel
in isolation;
it depends on which systems are released and which party is
trusted.
In the \emph{trusted Data Processor} setting, the Data Processor
is assumed to execute the prescribed learning algorithm and to release only
the specified output to the Investigator.
In the \emph{untrusted Data
Processor} setting, the Data Processor itself may be adversarial and may
attempt to extract more information from the Respondent's encoded data than
the prescribed output is meant to reveal.
The latter distinction is the motivation for the admissibility notion
introduced later in the paper: in the untrusted setting, one must certify not
only that the released output is private, but also that the prescribed
procedure is not merely a noisy post-processing of a strictly more
informative procedure that the Data Processor could have performed.

Our first contribution is a stability-to-generalization theorem for quantum
learning algorithms.  We use an information-theoretic stability measure based
on the dependence between the Respondent's data, the evaluation systems, and
the complete classical--quantum output of the learning procedure.  Under a
classical--quantum sub-Gaussian condition on the loss observables, we prove
that this stability controls the expected generalization error.

Conceptually, the theorem says that if the full output of the quantum
learning procedure contains little information about the particular training
instance, then the empirical loss and the true loss are close.  This extends
the mutual-information approach to classical generalization \cite{XR_2017}
to the setting of quantum learning with observable-valued losses and residual
quantum outputs \cite{Caro23,WDH2025}.  An important feature of the result is
that the loss assumption is imposed as a single classical--quantum
sub-Gaussian condition.  This differs from separated assumptions appearing
in earlier quantum generalization bounds, where classical sampling
fluctuations and quantum fluctuations are controlled separately
\cite{Caro23,WDH2025}.  Our formulation captures both effects in one
information-theoretic bound.

We also establish a high-probability generalization guarantee.  The expected
bound controls the average generalization error over the randomness of the
data and the learning procedure.  A high-probability bound is stronger in a
different direction: it controls the realized generalization error except on
a small exceptional event.  In the quantum setting, such a statement requires
tools that can handle higher-order information dependence in the presence of
non-commutativity.  For this reason, the proof uses sandwiched
Rényi-divergence techniques, in line with finite-resource quantum
information theory \cite{Tomamichel_2016,6574274} and with the Rényi-divergence
approach to quantum learning generalization in \cite{WDH2025}.  This result
complements the expected generalization theorem by providing a
concentration-type guarantee adapted to quantum learning models with
observable-valued losses.

The second part of the paper proves that quantum differential privacy
implies information-theoretic stability in the trusted Data Processor
setting.
We introduce a one-neighbor quantum differential privacy
condition for quantum learning algorithms.
The condition requires that
outputs produced from neighboring datasets be indistinguishable up to the
prescribed privacy parameters.
This requirement limits how much the
released classical-quantum output can depend on any individual data entry.
We prove a mechanism-independent upper bound on the corresponding
Holevo-type information between the training data and the released output.
The proof uses a quantum version of the type-covering and grid-covering
strategy underlying the classical analysis of differentially private
learning \cite{RGBS21}, together with standard quantum information
inequalities.

Combining this privacy-induced stability bound with the
stability-to-generalization theorem yields a privacy-to-generalization
guarantee for quantum learning algorithms.
Thus, in the trusted setting,
one-neighbor quantum differential privacy is not only a privacy constraint
but also a sufficient condition for information-theoretic stability, and
therefore for generalization.
In the purely classical and pure-DP limit,
our stability bound recovers the corresponding classical mutual-information
bound of \cite{RGBS21}.
In the approximate and quantum-output setting, it
extends that analysis by allowing a residual quantum system and by keeping
track of the additional overhead caused by approximate privacy.
We also
compare our stability bound with recent quantum local-privacy and
Holevo-information bounds \cite{Nuradha25,DWH2025,wilde2025private}, and
with the quantum learning generalization framework of \cite{Caro23}.

The trusted setting, however, does not address all privacy threats.  
If the
Data Processor is untrusted, output privacy alone is not sufficient.  An
adversarial processor could first perform a more informative quantum
procedure on the Respondent's encoded data, retain or use the information
thereby extracted, and only afterwards apply noisy post-processing to
produce an output that appears private to the Investigator.  
In such a case,
the privacy guarantee of the final released output would not certify that
the prescribed learning procedure itself is responsible for the privacy.  
It would only certify that information was discarded after a potentially more
revealing computation had already been performed.

This issue is naturally expressed in the language of post-processing order.
Classically, if the output of one experiment can be simulated from the output
of another by a Markov kernel, then the latter is at least as informative in
Blackwell's sense \cite{Blackwell51,BG54}.  In the quantum setting, the
corresponding comparison problem is more delicate: depending on the
operational formulation, one may compare models by CPTP maps, positive
trace-preserving maps, quantum statistical morphisms, or families of
decision problems \cite{Buscemi2012,Jencova2012,Matsumoto:2015vgh,Buscemi2016,BST2019}.
For privacy certification in an untrusted implementation, the relevant
question is not whether two statistical models are equivalent in a universal
sense, but whether the prescribed learning procedure is merely a
post-processing of another physically allowed procedure that extracts more
information from the same encoded ensemble.

This motivates the notion of \emph{Information-Theoretic Admissibility}
(ITA).  
ITA is not itself a privacy condition. Rather, it rules out the following
possibility: the prescribed procedure \(\mathcal N\) can be obtained by
applying a noisy post-processing map to another physically allowed procedure
\(\mathcal N'\) that extracts strictly more information from the same
encoded ensemble.
More explicitly, we say that a procedure \(\mathcal N'\) is more informative
than \(\mathcal N\) if the output of \(\mathcal N\) can be obtained by
applying an additional noisy post-processing step to the output of
\(\mathcal N'\).  In other words, \(\mathcal N\) does not discard
information by itself, but only after \(\mathcal N'\) has already extracted
it.
If such a strictly more informative \(\mathcal N'\) exists, then an
untrusted Data Processor could run \(\mathcal N'\), keep the additional
information, and only then apply \(\Gamma\) to produce the apparently private
output of \(\mathcal N\).  In that case, the privacy of the released output
does not certify privacy against the Data Processor.
A prescribed procedure is called ITA if no strictly more informative
procedure of this kind exists.  In other words, it is not merely a degraded
version of another physically allowed procedure that extracts more
information from the same data.
Thus, ITA rules out privacy claims that are achieved by
first extracting more information than necessary, and then
hiding this information through an additional noisy post-processing step.
This definition is
inspired by data-processing ideas and by post-processing orders used in the
comparison of classical and quantum statistical experiments
\cite{Blackwell51,Buscemi2012,Buscemi2016,DWH2025}, but it is applied here
in an ensemble-dependent form tailored to quantum learning.

The classical and quantum implications of ITA are markedly different.
In a classical, or more generally jointly commuting, model, an admissible procedure that is sufficiently informative tends to approach recoverability of the underlying classical dataset.
Hence admissibility and nontrivial
privacy are in strong tension: if the processor can perform an admissible
classical procedure, then privacy of the final output may not prevent the
processor from having learned essentially the raw data.
This classical
collapse shows why output privacy alone is inadequate in the untrusted
setting, and why an admissibility condition must be examined together with
the privacy claim.

In a genuinely quantum model, non-commutativity and non-orthogonality
change the situation.
Even an undominated physically allowed measurement
need not perfectly identify the underlying classical string, because
non-orthogonal quantum states cannot be perfectly discriminated.
This
limitation is not an added noise mechanism; it is a physical limitation on
state discrimination, closely related to the Helstrom bound
\cite{Helstrom1967_2}.
Thus, quantum mechanics does not by itself
guarantee privacy, but it can make admissibility compatible with
nontrivial privacy.
The essential point is that extracting all information
available to a prescribed quantum procedure need not coincide with
recovering the entire classical dataset.

We illustrate this distinction through a quantum ITA example based on
non-orthogonal encodings and a Hamming-weight measurement.
The purpose of
the example is not to propose a practical protocol, but to separate two
notions that coincide in purely classical models: being undominated in the
relevant post-processing order and perfectly recovering the underlying
classical string.
In the orthogonal limit, the corresponding measurement
has the character of a Quantum Non-Demolition measurement
\cite{Braginsky1980}.
Away from that limit, non-orthogonality prevents
perfect recovery even when the measurement exhausts the relevant
physically available information.
This example shows how nontrivial
privacy guarantees can remain meaningful against an untrusted Data
Processor when the privacy claim is paired with an admissibility condition.

The remainder of the paper is organized as follows.  Section
\ref{sec_def_not} fixes notation and recalls the type method used in the
privacy-to-stability analysis.
Section \ref{Qfw} introduces the quantum
learning framework, the observable-valued losses, and the
information-theoretic stability notion, and proves the expected and
high-probability stability-to-generalization results.
Section
\ref{priv_gen} studies the trusted Data Processor setting, defines the
one-neighbor quantum differential privacy condition, and proves that
privacy implies stability and hence generalization.
Section \ref{SSA}
turns to the untrusted Data Processor setting, introduces
Information-Theoretic Admissibility, and analyzes the classical and
quantum consequences of this condition.
The appendices contain technical
proofs, comparisons with prior bounds, and the detailed analysis of the
quantum ITA example.
\section{Notations}\label{sec_def_not}
Let $\cD(\cH)$ denote the set of density operators on a finite-dimensional Hilbert space $\cH$.
\paragraph{Method of Types.}
For a finite alphabet { $\cZ$} of size $d$, the \emph{type} of a sequence $s \in \cZ^n$ is the frequency vector ${\bf f} \in \bbN_0^d$ satisfying $\sum_i f_i = n$.
We denote the set of all types by $T_d^n$, and the \emph{type class} (the set of all sequences with type $\bf f$) by $T_{\bf f} \subset \cZ^n$.
Following \cite{RGBS21}, we define two sequences $s, \tilde{s} \in \cZ^n$ to be \emph{$k$-neighbors}, denoted as $s \overset{k}{\sim} \tilde{s}$, if their types satisfy
$
k = \frac{1}{2}\sum_{a\in\cT}\abs{f_a(s)-f_a(\tilde{s})}.
$
Note that $s \overset{0}{\sim} \tilde{s}$ implies the sequences are identical up to permutation.
Further, for any string $a \in \{0,1\}^n$, we denote $\abs{a}_1$ to be the hamming weight of $a$, i.e., the number of entries with the value $1$ in $a$.

\section{A General Framework for Quantum Learning Algorithms}
\label{Qfw}
\subsection{Learning Setup and Data Encoding}
In this section, we establish a quantum learning framework motivated by \cite{Caro23} and \cite{WDH2025}, operationalized through the interaction between a \textbf{Respondent} (data contributor) and a \textbf{Data Processor} (algorithm executor), \emph{focusing on the non-private setting} (i.e., without imposing any privacy constraint at this stage).

This section should be understood as the baseline mathematical model on which the privacy results of the later sections are built.
At this stage, we do not yet impose any adversarial constraint on the Data Processor, nor do we require that the learning map be private.
Accordingly, the family of instruments $\{\cN^{(s)}\}_{s\in\cS}$ is allowed to depend on the classical dataset label $s$.
This dependence is useful for formulating the most general learning procedure and for defining the joint classical-quantum state used in the generalization analysis.
The trusted and untrusted privacy restrictions introduced later will be obtained by imposing additional operational constraints on this baseline model.

The \textbf{Respondent} provides a classical dataset $s := (z_1,\dots,z_n) \in \cS$ (where $z_i=(x_i,y_i)$ maps input $x_i$ to label $y_i$), { drawn from a probability distribution $P_S$\footnote{{$P_S$ does not necessarily have an i.i.d. structure}}} and encoded into an aggregate quantum state $\rho_{s}:= \bigotimes_{i=1}^n \rho_{z_i} \in \cD(\cH^{\hat{\Te}} \otimes \cH^{\hat{\Tr}})$, where $\rho_{z_i}$ is the quantum state corresponding to $i$-th data $z_i$.
This state spans a training system $\Tr:= \hat{\Tr}^{\otimes n}$ (accessible to the Processor) and a testing system $\Te:=\hat{\Te}^{\otimes n}$ (used for evaluation).

The \textbf{Data Processor} receives the classical-quantum input $\sum_{s}P_S(s)\ketbra{s}\otimes \rho_s$ and executes a learning algorithm, modeled as a collection of quantum instruments $\cN := \{\cN^{(s)} : \cD(\cH^{\Tr}) \to \cD(\cH^{B})\}_{s \in \cS}$.
The output system $B \equiv WB'$ comprises a classical hypothesis $W$ and a quantum residue $B'$.
The distinction between $W$ and $B'$ is important. 
The classical register $W$ represents the hypothesis eventually used by the Investigator, while the quantum system $B'$ represents any residual quantum information retained or released by the learning procedure.
Even if the classical hypothesis $W$ appears stable, the residual system $B'$ may still contain information about the training data.
For this reason, our stability and generalization bounds are formulated for the full output $WB'$, rather than for $W$ alone.
The resulting joint state is given by:
\begin{equation}
    \sigma^{S\Te B}_{\cN} := \sum_{s \in \cS}P_S(s)\ketbra{s} \otimes (\sigma^{\cN}_s)^{\Te B},
\label{HJ2}
\end{equation}
where the output state conditioned on the Respondents' input $s$ is,
\begin{equation}
    \sigma^{\cN}_s:= \sum_{w \in \cW}((\bbI^{\Te} \otimes \cN^{(s)}_{w})(\rho_s))^{\Te B'} \otimes \ketbra{w}^W.
\label{sigma_s_state}
\end{equation}
 The action of the Data Processor's instrument $\cN^{(s)}$ is denoted as:
\begin{align}
\cN^{(s)}(\rho_s):=\sum_{w \in \cW} (\bbI^{\Te} \otimes  \cN^{(s)}_{w})(\rho_s)\otimes \ketbra{w}, \label{SHJ}
\end{align}
where each $\cN^{(s)}_{w}$ is a completely positive trace non-increasing map.
This interaction is illustrated in Figure \ref{fig:quantum_diag1}.

\begin{figure}[h]
\centering
\resizebox{100mm}{!}{
\begin{circuitikz}
\tikzstyle{every node}=[font=\large]
\draw [ font=\LARGE, color={rgb,255:red,255; green,0; blue,0} , fill= green!10, line width=1.3pt , rounded corners = 11.4] (4.25,10.5) rectangle  node[yshift=0.75cm] {$\cN = \left\{\cN^{(s)}\right\}_{s \in \cS}$} (9.25,7.25);
\draw [ color={rgb,255:red,255; green,0; blue,0} , fill= green!10, line width=1.3pt , rounded corners = 11.4] (-2.5,9.5) rectangle  node {\textbf{\textit{Respondent}} } (1.5,8.0);

\node [font=\large, color ={rgb,255:red,255; green,0; blue,0} ] at (-0.5,7.5)  {$\mathbf{\sum_{s}P_{S}(s) \ketbra{s} \otimes \rho_s}$ };


\draw [->, line width=1.3pt,color = cyan](1.5,8.75) to[short] node[color = black, xshift = -0.75cm, yshift= -0.25cm] {$\mathbf{\{s,\rho_s\}}$} (4.25,8.75);
\draw [->, dashed, line width=1.3pt,color = cyan](9.25,8.75) to[short] node[color = black, xshift = -1cm, yshift= -0.25cm] {$\mathbf{\cN^{(s)}(\rho_s)}$} (12.5,8.75);
\node [font=\large,  blue] at (6.75,8.9) {\textbf{\textit{Learning Algorithm}}};

\node [font=\large,  red] at (6.75,7.75) {\textbf{\textit{Data Processor}}};


\end{circuitikz}

}\caption{A general quantum learning framework.}\label{fig:quantum_diag1}
\end{figure}

\subsection{Stability of a Quantum Learning Algorithm}
We now define stability for quantum learning algorithms.
Intuitively, stability requires the learning outcome to remain invariant to \textit{single-entry} modifications, thereby preventing the leakage of individual data points.
Extending the classical information-theoretic framework of \cite{XR_2017}, which quantifies stability via mutual information, we formalize this notion below.
\begin{definition}(Stability)\label{def_stable_learn}
    A quantum learning algorithm $\cN = \left\{\cN^{(s)}\right\}_{s}$ is defined to be $\gamma$-stable, if $\max_{P_S} I\left[S \Te;WB'\right] \leq \gamma$, 
    where $I\left[S\Te;WB'\right]$ is calculated with respect to the classical-quantum state mentioned in \eqref{sigma_s_state}.
\end{definition}
The above definition provides a quantitative upper bound on the information
that can be extracted from the Data Processor's output system \(B\)
\((B \equiv WB')\) about the composite system consisting of the
Respondent's input dataset \(S\) and the testing system \(\Te\).
The appearance of both $S$ and $\Te$ in the information term reflects the fact that the learning procedure may create correlations not only with the classical dataset label but also with the quantum testing system associated with the encoded data.
Similarly, the full output $WB'$ is used because both the announced hypothesis and the residual quantum system may carry information about the training instance.
Thus, Definition~\ref{def_stable_learn} is deliberately stronger than a purely classical stability condition involving only $I[S;W]$.
It measures the total classical-quantum dependence between the Respondent's data and the Data Processor's released output.
Consequently, a small upper-bound implies that the algorithm’s output is not strongly dependent on any single training data point, indicating that the algorithm is information-theoretically stable.
\subsection{Stability Implies Generalizability For Quantum Learning Algorithms}
In this section, we demonstrate that if the Data Processor employs a stable algorithm, the results generalize well to unseen data.
For a quantum learning algorithm $\cN = \{\cN^{(s)}\}$, the joint state representing the Respondent's input and the Data Processor's output, mentioned in \eqref{HJ2}, can be expanded as:
\begin{equation}
    \sigma^{S\Te WB'}_{\cN} := \sum_{(s,w) \in \cS\times\cW}P_S(s)\ketbrasys{s}{S} \otimes P^{\cN}_{W\mid S} (w\mid s)\ketbrasys{w}{W} \otimes(\sigma^{\cN}_{s,w})^{\Te B'} ,
\label{HJ3}
\end{equation}
where $P^{\cN}_{W\mid S} (w\mid s):= \tr\left[(\bbI^{\Te} \otimes  \cN^{(s)}_{w})(\rho_s)\right]$ is the probability of the Data Processor selecting hypothesis $w$ given the dataset $s$, and $\sigma^{\cN}_{s,w}$ is the normalized residual state $\frac{(\bbI^{\Te} \otimes  \cN^{(s)}_{w})(\rho_s)}{P^{\cN}_{W\mid S} (w\mid s)}$.
In the following discussion, we define how to quantize the loss or error induced from the resultant state $\sigma^{\cN}$.

In the earlier discussed quantum learning framework, the input data $s$ and output hypothesis $w$ induced by the quantum learning algorithm $\cN$ are embedded into the output residue quantum state $\sigma^{\cN}_{s,w}$.
Therefore, to evaluate the performance of the Data Processor, we define the loss in terms of the expected value of observables with respect to the state $\sigma^{\cN}$ produced by the Data Processor.
In \cite{Caro23,WDH2025}, the authors consider a family of non-negative self-adjoint loss observables $\left\{L(s,w)\right\}_{\substack{(w,s)}}$ which act on the quantum testing system $\Te$ and the output quantum system $B'$.
Using these loss observables, we define the following global loss operator,
 \begin{equation}
 \label{GLS}
L^{S\Te WB'}:=   \sum_{(s,w) \in \cS\times\cW}\ketbrasys{s}{S} \otimes \ketbrasys{w}{W} \otimes \overset{\Te B'}{L(s,w)}.
\end{equation}

Based on the above description of the joint state $\sigma^{S\Te WB'}_{\cN}$ and the loss operators $\{L(s,w)\}$, we now distinguish between the loss observed by the Data Processor on the training data (empirical) and the loss expected on unseen fresh data (true).

The empirical loss is evaluated on the joint state generated by the actual training procedure.
It therefore retains the correlations between the dataset, the selected hypothesis, and the residual quantum output.
By contrast, the true loss is evaluated by breaking this dependence: the data used for evaluation are fresh and independent of the released output.
This is why the empirical loss is computed with $\sigma^{S\Te WB'}_{\cN}$, whereas the true loss is computed with the product state $\sigma^{S\Te}\otimes\sigma^{WB'}_{\cN}$.
The generalization error is precisely the discrepancy between these two evaluations.
\begin{definition}[Expected Empirical Loss {\cite[Definition $11$]{Caro23}}]\label{exp_emp_loss_def}
    The expected empirical loss $\hat{L}_{\rho} (\cN)$ captures the performance of the Data Processor's algorithm on the dataset provided by the Respondent.
It is the expectation over the joint distribution induced by the algorithm:
    \begin{align*}
        \hat{L}_{\rho} (\cN) &:= \bbE_{(S,W) \sim {P}^{\cN}_{SW}}[\tr[L(S,W)(\sigma^{\cN}_{S,W})^{\Te B'}]] = \tr[L^{S\Te WB'}\sigma^{S\Te WB'}_{\cN}].
\end{align*}
\end{definition}

\begin{definition}[Expected True Loss {\cite[Definition $19$]{WDH2025}}] \label{exp_true_loss_def}
    The expected true loss \newline $L_{\rho} (\cN)$ represents the generalization performance of the Data Processor's algorithm.
It evaluates the hypothesis $W$ generated by the Data Processor against a fresh dataset $\overline{S}$ independent of the training data $S$:
    \begin{align*}
        L_{\rho} (\cN)&:=\bbE_{(\overline{S},\overline{W}) \sim P_S \times {P}^{\cN}_{{W}} }\left[\tr\left[L(\overline{S},\overline{W})\left(\rho^{\Te}_{\overline{S}} \otimes (\sigma^{\cN}_{\overline{W}})^{B'}\right)\right]\right] =\tr[L^{S\Te WB'}(\sigma^{S\Te} \otimes \sigma^{WB'}_{\cN})].
\end{align*}
    where for any $s$, we define $\rho^{\Te}_{s}:= \tr_{\Tr}[\rho_s]$, for each $w$, we define  $\sigma^{\cN}_{w} := \bbE_{S \sim P^{\cN}_{S|W=w}} [\tr_{\Te}[\sigma^{\cN}_{S,w}]]$, and  $\sigma^{S\Te} $ and $ \sigma^{WB'}_{\cN}$ are the corresponding marginals of the state $\sigma^{S\Te WB'}_{\cN}$ defined in \eqref{HJ3}.
\end{definition}

\begin{remark}
    We adopt the definition of expected true loss as proposed in \cite{WDH2025} and not that of \cite{Caro23}.
The authors in \cite{WDH2025} give a rigorous justification for Definition \ref{exp_true_loss_def} and argue that the definition proposed by \cite[Definition $12$]{Caro23} is not a correct definition for the expected true loss.
\end{remark}

Operationally, this definition corresponds to the following experiment. 
The Data Processor first produces the output $WB'$ from the training data.
Then the performance of the resulting hypothesis is evaluated against an independently sampled fresh data system.
The product state $\sigma^{S\Te}\otimes\sigma^{WB'}_{\cN}$ encodes exactly this independence. 
Thus, the distinction between empirical and true loss is not a formal artifact;
it is the mathematical representation of evaluating a learned hypothesis on unseen data.

Based on these definitions, the expected generalization error is defined as the deviation between the Data Processor's empirical performance and the true performance.
\begin{definition}[Expected Generalization Error \cite{WDH2025}]\label{gen_ws_error_exp}
    The expected generalization error is:
    \begin{align*}
    \overline{\text{\textnormal{gen}}}_{\rho}(\cN) &:= \abs{ \hat{L}_{\rho}(\cN) - L_{\rho}(\cN)} = \abs{\tr[L^{S\Te WB}\sigma^{S\Te WB}_{\cN}] - \tr\left[L^{S\Te WB}(\sigma^{S\Te} \otimes \sigma^{WB'}_{\cN})\right]}.
\end{align*}
\end{definition}

We will now bound $\overline{\text{\textnormal{gen}}}_{\rho}(\cN)$ in terms of $I\left[S \Te;WB'\right].$ To obtain such a bound in the classical setting \cite{XR_2017} assumed that the loss function is sub-Gaussian.
We will make a similar assumption for the loss operators $\{L(w,s)\}$ and the Data Processor's output state.
Towards this, we make the following definition.  

The role of the following condition is to control the fluctuations of the loss observable under the product reference state.
It is the quantum analogue of the classical sub-Gaussian tail assumption used in information-theoretic generalization bounds.
Here the randomness has two sources: the classical randomness of $(S,W)$ and the quantum uncertainty associated with measuring the loss observable on $\Te B'$.
The condition below packages these two sources into a single variance proxy $\alpha^2$.
\begin{definition}[ Classical-Quantum $\alpha$-Sub-Gaussianity]\label{ass_cq_subg}
For a fixed parameter $\alpha \in (0,\infty),$ the collection $\{L(w,s)\}$ of loss operators  is said to be an $\alpha$-sub-Gaussian collection with respect to $\sigma^{S\Te} \otimes \sigma^{WB'}_{\cN} := \sum_{(s,w) \in \cS\times\cW}P_S(s)\ketbra{s} \otimes P^{\cN}_{W} (w\mid s)\ketbra{w}\otimes\rho^{\Te}_{s} \otimes  (\sigma^{\cN}_{w})^{B'},$
    if for every $\lambda \in \bbR,$ it satisfies,
    \begin{equation}
\bbE\left[\tr\left[e^{\lambda\left(L(S,W) - \bbE\left[\tr\left[L(S,W) \left(\rho^{\Te}_{S} \otimes (\sigma^{\cN}_{W})^{B'}\right)\right]\bbI^{\Te B'}\right]\right)}\left(\rho^{\Te}_{S} \otimes (\sigma^{\cN}_{W})^{B'}\right)\right]\right] \leq e^{\frac{\lambda^2 \alpha^2}{2}},\label{ass_cq_subg_eq}
    \end{equation}
    where the expectations are calculated with respect to the product distribution $ P_S \times {P}^{\cN}_{{W}}$.
Note that \eqref{ass_cq_subg_eq} is equivalent to,
    \begin{equation}
        \tr\left[e^{\lambda\left(L^{S\Te WB} - \tr\left[L^{S\Te WB}(\sigma^{S \Te} \otimes \sigma^{WB'}_{\cN}) \right]\bbI^{S\Te WB}\right)}(\sigma^{S \Te} \otimes \sigma^{WB'}_{\cN})\right] \leq e^{\frac{\lambda^2 \alpha^2}{2}},\label{ass_cq_subg_eq2}
    \end{equation}
    where $L^{S\Te WB'}$ is the global loss operator defined in \eqref{GLS}.
\end{definition}

{ Definition \ref{ass_cq_subg} naturally generalizes classical sub-Gaussianity. In the limit of trivial quantum systems ($\dim(\Te) = \dim(B')=1$), the operators $\sigma^{\cN}_{W}$ and $\bbI^{B'}$ become scalars, reducing $L(S,W)$ to a classical random loss function.
Consequently, condition \eqref{ass_cq_subg_eq} collapses to the standard classical sub-Gaussian inequality $\bbE_{(S,W)}[e^{\lambda(L(S,W) - \bbE[L(S,W)])}] \leq e^{\frac{\lambda^2 \alpha^2}{2}}$ with respect to $P_S \times {P}^{\cN}_{{W}}$.
We now present a theorem bounding the expected generalization error for quantum learning algorithms.}

\begin{theorem}\label{exp_gen_bound}
   For a fixed $\alpha \in (0,\infty),$ if the loss operators for a quantum learning algorithm $\cN$, satisfy Definition \ref{ass_cq_subg}, then, we have,
     \begin{equation}
         \overline{\text{\textnormal{gen}}}_{\rho}(\cN) \leq \sqrt{2\alpha^2 I[S \Te;WB']}.
\label{gen_bound_eq}
    \end{equation}
\end{theorem}

\begin{proof}
    See Appendix \ref{proof_exp_gen_bound} for the proof.
\end{proof}

The proof is based on a transport-type argument. 
The relative entropy 
\[
D(\sigma^{S\Te WB'}_{\cN}\|
\sigma^{S\Te}\otimes\sigma^{WB'}_{\cN})
= I[S\Te;WB']
\]
measures how far the actual joint state is from the product state corresponding to fresh evaluation data.
The classical-quantum sub-Gaussian condition converts this relative-entropy distance into a bound on the difference of expected losses.
Thus, the theorem states that generalization follows whenever the learning output is nearly independent of the training data in the information-theoretic sense.

The following corollary of Theorem \ref{exp_gen_bound}, together with Definition \ref{def_stable_learn}, indicates that when the Data Processor employs a stable algorithm, the generalization error remains tightly bounded.
\begin{corollary}
\label{SimpliesG}
     If the Data Processor's learning algorithm $\cN:=\{\cN^{(s)}\}$ is $\gamma$-stable and the loss operators of $\cN$ satisfy Definition \ref{ass_cq_subg}, for a fixed $\alpha \in (0,\infty),$ then, its generalization error is upper bounded by $\sqrt{2\alpha^2 \gamma}$.
\end{corollary}

\subsection{Probabilistic Bounds and True Loss Lower Bound}
While bounds on the expected generalization error provide a measure of average performance, robust learning requires guarantees that hold with high confidence for individual realizations of the algorithm.  To address this dependence, we first introduce the Sandwiched R\'enyi divergence, the i.i.d. structure of the data, and the relevant sub-Gaussianity and error definitions.

The Sandwiched R\'enyi divergence \cite{MDSF_Renyi_2013} of order $\gamma \in (1, \infty)$ for two quantum states $\rho$ and $\sigma$ is defined as:
\begin{equation}
\Tilde{D}_{\gamma}(\rho \| \sigma) := \begin{cases}
            \frac{1}{\gamma - 1}\log\tr\left[\left(\sigma^{\frac{1 - \gamma}{2\gamma}}\rho\sigma^{\frac{1 - \gamma}{2\gamma}}\right)^{\gamma}\right], & {\mbox{if }} (\rho \ll \sigma),\\
            +\infty, & \mbox{else}.
\end{cases}\label{def_sand_renyi}
\end{equation}

\paragraph{I.I.D. Structure of Data and Algorithm.}
We assume the dataset $S = \{Z_1, \dots, Z_n\}$ consists of $n$ i.i.d. random variables, where each $Z_i \sim P_Z$. Commensurate with this, we assume the quantum learning algorithm $\cN$ respects this independence by acting on each data encoding locally. Specifically, the global channel decomposes as a tensor product,
\[
\cN^{(S)}_{w} := \bigotimes_{i=1}^{n} \cN^{(Z_i)}_{w},
\]
where each local map $\cN^{(Z_i)}_{w}: \cH^{\hat\Tr} \to \cH^{\hat B}$ acts on the input state $\rho_{Z_i}$ corresponding to the $i$-th datapoint. Consequently, the residual quantum output system decomposes as $B' := \hat{B}^{\otimes n}$.

\paragraph{Decomposition of Loss.}
Consistent with the independence of the processing, we assume the global loss observable $L(w,s)$ is the average of local loss observables acting on the individual subsystems:
\begin{align}
     L(w,s) := \frac{1}{n}\sum_{i=1}^{n}(\bbI^{\hat{\Te}} \otimes \bbI^{\hat{B}})^{\otimes (i-1)} \otimes \hat{L}(w,z_i) \otimes (\bbI^{\hat{\Te}} \otimes \bbI^{\hat{B}})^{\otimes( n-i)},\label{loss_observables_frag}
\end{align}
where $\hat{L}(w,z_i)$ is the local loss observable for the $i$-th data point. To guarantee exponential probability, we require the local loss operators to satisfy a sub-Gaussian condition.

\begin{definition}[Classical-Quantum Local $\alpha$-Sub-Gaussianity]\label{ass_cq_subg_loc}
For a fixed parameter $\alpha \in (0,\infty),$ the collection $\{\hat{L}(w,z)\}$ of local loss operators is said to be an $\alpha$-sub-Gaussian collection if, for every $\lambda \in \bbR$, the centered local moment generating function satisfies,
    \begin{equation}
\bbE\left[\tr\left[e^{\lambda\left(\hat{L}(Z_i,W) - \bbE\left[\tr\left[\hat{L}(Z_i,W) \left(\rho^{\Te}_{Z_i} \otimes (\sigma^{\cN}_{W})^{\hat{B}}\right)\right]\bbI^{\hat{\Te} \hat{B}}\right]\right)}\left(\rho^{\hat{\Te}}_{Z_i} \otimes (\sigma^{\cN}_{W})^{\hat{B}}\right)\right]\right] \leq e^{\frac{\lambda^2 \alpha^2}{2}},\label{ass_cq_subg_loc_eq}
    \end{equation}
    where the expectations are taken with respect to the product distribution $P_Z \times {P}^{\cN}_{{W}}$.
\end{definition}

\begin{definition}[Conditional Classical-Quantum Local $\alpha$-Sub-Gaussianity]\label{ass_cq_subg_loc_cond}
For a fixed parameter $\alpha \in (0,\infty),$ for every $w \in \cW$, the collection $\{\hat{L}(w,z)\}_{z}$ of local loss operators is said to be an $\alpha$-sub-Gaussian collection if, for every $\lambda \in \bbR$, the centered local moment generating function satisfies,
    \begin{equation}
\bbE\left[\tr\left[e^{\lambda\left(\hat{L}(Z,w) - \bbE\left[\tr\left[\hat{L}(Z,w) \left(\rho^{\Te}_{Z} \otimes (\sigma^{\cN}_{w})^{\hat{B}}\right)\right]\bbI^{\hat{\Te} \hat{B}}\right]\right)}\left(\rho^{\hat{\Te}}_{Z} \otimes (\sigma^{\cN}_{w})^{\hat{B}}\right)\right]\right] \leq e^{\frac{\lambda^2 \alpha^2}{2}},\label{ass_cq_subg_loc_cond_eq}
    \end{equation}
    where the expectations are taken with respect to the distribution $P_{Z}$.
\end{definition}

Under this i.i.d. setting, we formally define the random variable representing the absolute deviation of the generalization error.

\begin{definition}[Absolute Generalization Error Deviation]
\label{def:gen_error_deviation}
Let $\cN$ be a quantum learning algorithm with the i.i.d. structure defined above. Using \cite[Definition $20$]{WDH2025}, for a given $w \in \cW$, we define the generalization error random variable as the absolute difference between the empirical loss and the true loss $L_{\rho}(\cN, w)$, (see \cite[Definition $17$]{WDH2025}),
\begin{align}
    &~~~\textnormal{gen}_{\rho}(\cN,S,w) \nn\\
    &:= \left| \tr\left[L(S,w)(\sigma^{\cN}_{S,w})^{\Te B'}\right] - L_{\rho}(\cN,w) \right|\nn\\
    &= \abs{\frac{1}{n} \sum_{i=1}^{n}\tr\left[\hat{L}(Z_i,w)(\sigma^{\cN}_{Z_i,w})^{\hat{\Te} \hat{B}}\right] - \bbE_{\overline{Z}\sim P_Z }\left[\tr\left[\hat{L}(\overline{Z},w)\left(\rho^{\hat{\Te}}_{\overline{Z}} \otimes (\sigma^{\cN}_{w})^{\hat{B}}\right)\right]\right]}.\nn
\end{align}
\end{definition}

With these definitions in place, we prove a quantum version of \cite[Corollary $2$]{Esposito21} in term of the sandwiched R\'enyi divergence, derived under the assumption of i.i.d. data and loss observable decompositions.

\begin{theorem}\label{thm:renyi_gen_bound}
Let $\cN$ be a quantum learning algorithm.
Assume that the associated collection of loss operators $\{L(s,w)\}$ satisfies the Conditional Classical-Quantum Local Sub-Gaussian condition (Definition \ref{ass_cq_subg_loc_cond}).
For any Sandwiched R\'enyi divergence order $\gamma > 1$ and confidence level $\delta \in (0,1)$, the generalization error is bounded with probability at least $1-\delta$ as:
\begin{equation}
\Pr_{(S,W) \sim P^{\cN}_{SW}}\left\{\textnormal{gen}_{\rho}(\cN,S,W) \leq \sqrt{\frac{2\alpha^2}{n} \left( \Tilde{D}_{\gamma} (\sigma^{S\Te WB'}_{\cN}\| \sigma^{S\Te} \otimes \sigma^{WB'}_{\cN}) + \frac{\gamma}{\gamma-1} \ln \frac{2}{\delta} \right)}\right\} \geq 1-\delta, \nn
\end{equation}
where $\textnormal{gen}_{\rho}(\cN,S,W)$ is the generalization error random variable (Definition \ref{def:gen_error_deviation}) and $\Tilde{D}_{\gamma}$ denotes the Sandwiched R\'enyi divergence (defined in \eqref{def_sand_renyi}).
\end{theorem}

\begin{proof}
    See Appendix \ref{app:renyi_bound} for the complete proof.
\end{proof}

Complementing these upper bounds in expectation and probability, we provide the following lower bound on the expected true loss in terms of the expected empirical loss.
\begin{theorem}
\label{thm:true_loss_lb}
Let $\cN$ be a quantum learning algorithm with loss operators satisfying the Classical-Quantum Sub-Gaussian property (Definition \ref{ass_cq_subg}) with parameter $\alpha > 0$.
For any sandwiched R\'enyi divergence order $\gamma > 1$, the expected true loss is lower bounded by the empirical loss in the following exponential form,
\begin{equation}
    \exp\left(L_{\rho} (\cN)\right) \geq \hat{L}_{\rho}(\cN) \exp\left( - \left[ \frac{\gamma \alpha^2}{2(\gamma-1)} + \frac{\gamma-1}{\gamma} \Tilde{D}_{\gamma} (\sigma^{S\Te WB'}_{\cN}\| \sigma^{S\Te} \otimes \sigma^{WB'}_{\cN}) \right] \right).
\end{equation}
\end{theorem}

\begin{proof}
    See Appendix \ref{app:true_los_lb} for the proof.
\end{proof}

In Appendix \ref{App:comp_gen}, we compare our upper-bounds on generalization error (Theorem \ref{exp_gen_bound} and Theorem \ref{thm:renyi_gen_bound}) with prior works.

\section{Generalization Guarantees for Differentially Private Quantum Learning}\label{priv_gen}
\subsection{Trusted Setting and One-Neighbor ($\eps,\delta$)-DP}
This section examines the connection between privacy and generalization in quantum learning.
Building on Section \ref{Qfw}, which linked information-theoretic stability to generalization, we demonstrate that differential privacy enforces this stability.
Using the framework of Figure \ref{fig:quantum_diag}, we introduce the \textbf{Investigator} as the recipient of the output system $B$, generated by a \textbf{Trusted Data Processor} from the Respondent's raw data ($S, \Te, \Tr$).

To prevent the reconstruction of individual entries, the Processor ensures the algorithm satisfies differential privacy, requiring output invariance under single-entry modifications.
This constraint is mathematically equivalent to algorithmic stability (Definition \ref{def_stable_learn}), confirming privacy as a sufficient condition for generalization.
We formalize this indistinguishability requirement below.

The trusted setting should be distinguished from the untrusted setting studied in Section~\ref{SSA}.
Here the Data Processor is assumed to execute the prescribed privacy-preserving algorithm.
Therefore, the adversarial party is the Investigator, who receives only the released system $B=WB'$.
The privacy requirement is consequently imposed on the output states produced by neighboring datasets.
In this model, differential privacy is a statement about the indistinguishability of the released output, not about limiting the internal access of the Data Processor.
\begin{figure}[h]
\centering
\resizebox{120mm}{!}{
\begin{circuitikz}
\tikzstyle{every node}=[font=\large]
\draw [ font=\LARGE, color={rgb,255:red,255; green,0; blue,0} , fill= green!10, line width=1.3pt , rounded corners = 11.4] (4.25,10.5) rectangle  node[yshift=0.75cm] {$\cN = \left\{\cN^{(s)}\right\}_{s \in \cS}$} (9.25,7.25);
\draw [ color={rgb,255:red,255; green,0; blue,0} , fill= green!10, line width=1.3pt , rounded corners = 11.4] (-3.5,9.5) rectangle  node {\textbf{\textit{Respondent}} } (0.5,8.0);
\node [font=\large, color ={rgb,255:red,255; green,0; blue,0} ] at (-1.5,7.5)  {$\mathbf{\sum_{s}P_{S}(s) \ketbra{s} \otimes \rho_s}$ };
\draw [ color={rgb,255:red,255; green,0; blue,0} , fill= green!10, line width=1.3pt , rounded corners = 11.4] (13.75,9.5) rectangle  node {\textbf{\textit{Investigator}}} (17.75,8.0);
\node [font=\large,  black] at (15.8,8.3) {\textbf{can access  B}};
\draw [->, line width=1.3pt,color = cyan](0.5,8.75) to[short] (4.25,8.75);
\draw [->, dashed, line width=1.3pt,color = cyan](9.25,8.75) to[short] node[color = black, xshift = -1cm, yshift= -0.25cm] {output system $B$} (13.75,8.75);
\node [font=\large,  blue] at (6.75,8.9) {\textbf{\textit{Learning Algorithm}}};
\node [font=\large,  red] at (6.75,7.75) {\textbf{\textit{Trusted Data Processor}}};

\node [font=\large, color = black] at (2,8.5) {$\mathbf{\{s,\rho_s\}}$};
\draw[<->, thick, dashed, color = black] (15.75,9.5) to[bend right=30]  (-1.5,9.5);
\node[font=\large, color = red] at (7,12.5) {\textbf{Investigator attempts to learn $S$}};
\end{circuitikz}
}\caption{Privacy based learning framework.}\label{fig:quantum_diag}
\end{figure}

The definition contains three components.
Permutation invariance removes dependence on the ordering of the samples and ensures that the algorithm depends only on the empirical type of the dataset.
The privacy condition is the quantum indistinguishability condition for outputs generated by neighboring datasets.
Support consistency is a technical regularity condition needed to control relative-entropy quantities in the approximate privacy regime $\delta>0$;
in the pure case $\delta=0$, it is automatically satisfied.

\begin{definition}\label{DPL}
An algorithm $\cN = \left\{\cN^{(s)}\right\}_{s \in \cS}$ is a $1$-neighbor $(\eps,\delta)$-DP support consistent quantum learning algorithm if it satisfies the following conditions:

\begin{enumerate}
    \item \textbf{Permutation Invariance:} For all $s,s' \in \cS$ satisfying $T_s=T_{s'}$, 
    the algorithm satisfies the condition $\cN^{(s)}(\rho_s) =
    \cN^{(s')}(\rho_{s'})$.
This ensures that the algorithm's output depends solely on the frequency of the data, not its specific ordering.
This condition is natural in statistical learning, where the order of training examples is irrelevant to the hypothesis, and it further adds an extra layer of privacy.
    \item \textbf{Privacy:} For every $s\overset{1}\sim s'$ and $0\preceq \Lambda \preceq \mathbb{I}$, the following inequality holds:
\begin{equation}
    \begin{split}
        \tr[\Lambda \cN^{(s)}(\rho_s)] &\leq e^\eps\tr[\Lambda \cN^{(s')}(\rho_{s'})] +\delta,\\
        \tr[\Lambda \cN^{(s')}(\rho_{s'})] &\leq e^\eps\tr[\Lambda \cN^{(s)}(\rho_s)] +\delta.
\end{split}
    \label{pp4}\end{equation}

    \item \textbf{Support Consistency:} For every $s\overset{1}\sim s'$, the output supports are identical, i.e.,
    \begin{equation}
    \supp(\cN^{(s)}(\rho_{s})) = \supp(\cN^{(s')}(\rho_{s'})).\label{pp3}
    \end{equation}
    See Remark \ref{rem_supp_con} for more details.
\end{enumerate}
\end{definition}

\begin{remark}\label{rem_supp_con}
The support consistency condition \eqref{pp3} is automatically satisfied in the pure differential privacy regime ($\delta=0$).
\end{remark}

\subsection{Privacy Implies Stability: A Mutual-Information Bound}
Definition \ref{DPL} above implies that privacy guarantees extend to $k$-neighbors, albeit with degraded parameters.
\begin{corollary}\label{fact_QDP_group}
    If $\cN$ satisfies Definition \ref{DPL}, then for any inputs $s \overset{k}{\sim} s'$ ($k\geq1$) and $0 \preceq \Lambda \preceq \bbI$, we have, $\tr[\Lambda \cN^{(s)}(\rho_s)] \leq e^{k \eps} \tr[\Lambda \cN^{(s')}(\rho_{s'})]+ g_{k}(\eps,\delta),$
    where $g_{k}(\eps,\delta) := \frac{e^{k\eps}-1} {e^{\eps}-1}\delta$ is assumed to be strictly less than $1$.
The symmetric inequality holds by swapping $s$ and $s'$.
\end{corollary}
\begin{proof}
    See Appendix \ref{fact_QDP_group_proof} for the proof.
\end{proof}

We now utilize the framework established in Section \ref{Qfw} to analyze the stability of quantum learning algorithms that satisfy Definition \ref{DPL}.
For this analysis, we modify the framework by treating the quantum test data system $\Te$ as trivial (i.e., $\dim(\Te) = 1$).

Under this modification, the stability measure from Definition \ref{def_stable_learn} simplifies to the mutual information between the training data $S$ and the output system $WB'$.
Therefore, in the theorem below, we derive an upper bound on $I[S;WB']$ for a quantum $(\eps, \delta)$-differentially private (DP) learning algorithm.
This derivation relies on the following assumption regarding the noise parameters $\eps$ and $\delta$,
\begin{equation} \label{ass1}
    g_{n(|\mathcal{Z}|-1)}(\eps,\delta) < 1,
\end{equation}
where for any $k\geq1$, $g_k(\eps,\delta)$ is defined in Corollary \ref{fact_QDP_group}.
\begin{theorem}\label{lemma_Holevo_stability}
For $\eps \in \left[\frac{1}{n},1)\right.$, consider $\cN = \left\{\cN^{(s)}\right\}_{s \in \cS}$ to be a  $1$-neighbor $(\eps,\delta)$-DP support consistent learning algorithm
 (see Definition \ref{DPL}) and satisfies the condition \eqref{ass1}.
Then, the following holds,
    \begin{equation}
        I[S;WB'] \leq (|\mathcal{Z}| - 1) \ln\left({n e\eps}\right) + h_{\abs{\cZ}}(\eps,\delta), \label{lemma_Holevo_stability_eq}
    \end{equation}
     where, $n$ is the length of the training data and for some constant $m\in(0,1],$ $h_{\abs{\cZ}}(\eps,\delta) := \ln \frac{1}{1- g_{{n(\abs{\cZ}-1)}}(\eps,\delta)} + \frac{2}{m}g_{n(\abs{\cZ}-1)}(\eps,\delta)$ and has a property that 
     $h_{\abs{\cZ}}(\eps,0) = 0$.
\end{theorem}

\begin{proof}
    See Appendix \ref{proof_theorem_main} for the proof.
\end{proof}

The proof has a simple structure. 
First, the mutual information $I[S;WB']$ is upper bounded by choosing a suitable reference output state $\omega^B$. 
Second, the set 
of empirical types of length $n$ is covered by a smaller grid. 
For each dataset, one compares its output state with the output state associated with a nearby grid representative. 
The differential privacy condition controls the divergence between these nearby outputs, while the number of grid points contributes the logarithmic covering term.
Optimizing the grid size yields the term $(|\cZ|-1)\ln(ne\eps)$, and the approximate privacy parameter $\delta$ contributes the overhead $h_{|\cZ|}(\eps,\delta)$.

The stability results for the case when $\eps \in \left.[0,\frac{1}{n}\right)$ and the case when $\eps \in \left(1\right.,\infty)$ follow from the proof techniques of Theorem \ref{lemma_Holevo_stability}.
We mention them as the corollaries below,

\begin{corollary}\label{corr_Holevo_stability_2}
    For $\eps \in \left.[0,\frac{1}{n}\right)$, consider a learning algorithm $\cN = \left\{\cN^{(s)}\right\}_{s \in \cS}$, which satisfies the properties mentioned in Definition \ref{DPL} and the condition \eqref{ass1}.
Then, $
        I[S;WB'] \leq (|\mathcal{Z}| - 1)\eps n + h_{\abs{\cZ}}(\eps,\delta).$
\end{corollary}

\begin{proof}
    See {Appendix \ref{proof_corr1}} for the proof.
\end{proof}

\begin{corollary}\label{corr_Holevo_stability_1}
    For $\eps \in \left(1\right.,\infty)$, consider a  learning algorithm $\cN = \left\{\cN^{(s)}\right\}_{s \in \cS}$, which satisfies the properties mentioned in Definition \ref{DPL} and the condition \eqref{ass1}.
Then, $
        I[S;WB'] \leq (|\mathcal{Z}| - 1) \ln\left(n+1\right).$
\end{corollary}

\begin{proof}
    See Appendix \ref{proof_corr2} for the proof.
\end{proof}

{ Theorem \ref{lemma_Holevo_stability} quantitatively links differential privacy to algorithmic stability by bounding the mutual information $I[S;WB']$ between the training data and the algorithm's output.
This bound is uniform and scales explicitly with dataset size $n$, alphabet size $|\cZ|$, and privacy parameters $(\eps,\delta)$.
By translating $(\eps,\delta)$-DP guarantees into a provable stability bound, the theorem establishes a direct connection between privacy and stability-based generalization controls.

Furthermore, Theorem \ref{lemma_Holevo_stability} strictly generalizes \cite[Proposition 2]{RGBS21}: by taking a trivial quantum system ($\dim(B')=1$) and setting $\delta=0$, the overhead $h_{\abs{\cZ}}$ vanishes, recovering the classical stability bound $(|\mathcal{Z}| - 1) \ln(n e\eps)$.}

\begin{remark}
    The upper-bound obtained in Theorem \ref{lemma_Holevo_stability} is independent of $P_S$ and thus, Theorem \ref{lemma_Holevo_stability} implies that if a quantum learning algorithm $\cN = \left\{\cN^{(s)}\right\}_{s \in \cS}$ satisfies Definition \ref{DPL}, then $\cN$ is $\left((|\mathcal{Z}| - 1) \ln\left({n e\eps}\right) +\right. $ $\left.h_{\abs{\cZ}}(\eps,\delta)\right)$-stable (see Definition \ref{def_stable_learn}).
A similar observation also follows for Corollaries \ref{corr_Holevo_stability_2} and \ref{corr_Holevo_stability_1}.
\end{remark}

In Appendix \ref{comp_stab}, we present a detailed comparison of Theorem \ref{lemma_Holevo_stability} with \cite[Proposition~10]{Nuradha25}.
Additionally, in Appendix \ref{comp_stab_car}, we provide an in-depth comparison of Theorem \ref{lemma_Holevo_stability} with the results of \cite[Appendix C.7]{Caro23} in the setting of untrusted Data Processors, a topic we will elaborate on in the subsequent section.
\subsection{From Stability to Generalization: DP Generalization Guarantees}
We now formally demonstrate that differential privacy guarantees generalization.
Conceptually, the argument is a two-step implication:
\[
\text{differential privacy}
\quad\Longrightarrow\quad
\text{information-theoretic stability}
\quad\Longrightarrow\quad
\text{generalization}.
\]
The first implication is Theorem~\ref{lemma_Holevo_stability}, which bounds $I[S;WB']$ using the privacy parameters.
The second implication is Theorem~\ref{exp_gen_bound}, which converts this information bound into a generalization bound.
The following corollary records the combined consequence.

By combining Theorem \ref{lemma_Holevo_stability} with Theorem \ref{exp_gen_bound} (assuming a trivial system $\Te$), which bounds the expected generalization error via the square root of mutual information, we establish a direct link.
Specifically, a $1$-neighbor $(\eps,\delta)$-DP support consistent algorithm limits dependence on individual data points, thereby ensuring robust generalization.
We formalize this result in the corollary below.

\begin{corollary}
    \label{cor:dp_generalization_bound}
    Consider a quantum learning algorithm $\cN$ that is $1$-neighbor $(\eps,\delta)$-DP support consistent (see Definition \ref{DPL}) with $\eps \in [\frac{1}{n}, 1)$ and satisfies condition \eqref{ass1}.
If the loss operator satisfies \eqref{ass_cq_subg_eq} (or equivalently \eqref{ass_cq_subg_eq2}) mentioned in Definition \ref{ass_cq_subg}, for some $\alpha \in (0,\infty)$, then, $\overline{\text{\textnormal{gen}}}_{\rho}(\cN) \leq \sqrt{2\alpha^2 \mathcal{I}_{bound}},$
    where $\mathcal{I}_{bound} = \left[ (|\mathcal{Z}| - 1) \ln\left({n e\eps}\right) + h_{\abs{\cZ}}(\eps,\delta) \right]$ is the upper bound on the mutual information derived in Theorem \ref{lemma_Holevo_stability}.
\end{corollary}

\section{Untrusted Data Processor and Information Theoretic Admissibility (ITA)}\label{SSA}
\subsection{Untrusted Data Processor Model}
In the previous sections, we assumed a \textit{trusted} Data Processor model where the Data Processor reliably executes the privacy-preserving algorithm and releases only the privatized output.
We now relax this assumption to address the Untrusted Data Processor scenario.
Here, the Data Processor is considered adversarial and may attempt to leak or extract information about the training data $s$ beyond what is contained in the intended output system $B$.
To address this privacy threat rigorously, we adopt a worst-case security model where the Data Processor and the Investigator collude or effectively act as a single adversarial entity.

This change of adversarial model is substantial. 
In the trusted setting, it is meaningful to define privacy only for the released output, because the Data Processor is assumed to follow the prescribed algorithm.
In the untrusted setting, however, the prescribed algorithm need not be the algorithm actually executed.
Therefore, an output-based privacy guarantee alone cannot certify privacy against the Data Processor.
The privacy mechanism must instead be evaluated together with the information that the encoded quantum states make physically accessible.

In this setting, the Respondent does not grant the Data Processor access to the raw classical data $s$ directly.
Instead, the Respondent provides access only to the set of encoded quantum states $\{\rho_s\}_s$.
The Processor is tasked with running a learning algorithm to produce a hypothesis $w$.
Since the Data Processor is untrusted, the learning algorithm is modeled as a single, fixed completely positive, trace-preserving (CP-TP) map $\mathcal{N}$ that must be independent of the specific input index $s$.
The total state generated at the output of this process is, $
\cN(\rho_s) := \sum_{w \in \cW} \cN_{w}(\rho_s) \otimes \ketbra{w},$
where each $\mathcal{N}_{w}$ is a completely positive trace non-increasing map summing to $\cN$.
To ensure privacy, the Respondent mandates that this map $\mathcal{N}$ must satisfy differential privacy constraints with respect to neighboring inputs.
We formalize this via the following definition for an untrusted Data Processor.

\begin{definition}\label{DPL2}An algorithm $\cN$ is said to be a $1$-neighbor $(\varepsilon,\delta)$-DP support-consistent learning algorithm with an untrusted Data Processor if the following conditions hold.
\begin{enumerate}
\item \textbf{Permutation Invariance:}
For all $s,s' \in \cS$ satisfying $T_s=T_{s'}$, 
    the algorithm satisfies the condition $\cN(\rho_s) = \cN(\rho_{s'})$.
\item \textbf{Privacy:} For every $s \overset{1}{\sim} s'$ (see Section \ref{sec_def_not}) and $0 \preceq \Lambda \preceq \mathbb{I}$, the following inequality holds:
\begin{align}
\tr[\Lambda \cN(\rho_s)] &\leq e^\varepsilon \tr[\Lambda \cN(\rho_{s'})] + \delta,\nn\\
\tr[\Lambda \cN(\rho_{s'})] &\leq e^\varepsilon \tr[\Lambda \cN(\rho_{s})] + \delta.
\nn\end{align}
\item \textbf{Support Consistency:} For every $s \overset{1}{\sim} s'$, the output supports are identical, i.e.,\begin{equation}\supp(\cN(\rho_s)) = \supp(\cN(\rho_{s'})).
\label{pp3A}\end{equation}\end{enumerate}
\end{definition}

\subsection{Information-Theoretic Admissibility (ITA): Motivation and Definition}

While the Respondent prescribes a specific privacy-preserving algorithm
\(\mathcal{N}\), an adversarial Data Processor possessing the raw quantum
inputs \(\{\rho_s\}_s\) is not technically bound to execute
\(\mathcal{N}\).
The Data Processor aims to extract the maximum possible
information about \(s\).
Therefore, there exists a risk that the Data
Processor might execute a \textit{strictly more informative} algorithm
\(\mathcal{N}'\) and then perform post-processing so that the publicly
visible behavior is indistinguishable from that of the prescribed
algorithm \(\mathcal{N}\).
In operational terms, the Data Processor
could first perform a non-private learning procedure and only afterwards
artificially stitch up the noise needed to reproduce the apparent output
of the prescribed private algorithm.

If such a scenario is possible, the privacy guarantees calculated based
on \(\mathcal{N}\) (as mentioned in Definition \ref{DPL2}) are rendered
void, as the Data Processor effectively holds the information content of
\(\mathcal{N}'\).
To formalize this obstruction, we use an ordering of informativeness between learning algorithms.

\begin{definition}[Information ordering of Algorithms]\label{def_inf_chan}
    Let $\cN:=\{\cN_w\}_w$ and $\cN':=\{\cN'_w\}_w$ be two quantum learning algorithms and
$\{\rho_s\}_s$ be a fixed set of input states to these algorithms.
We say that $\cN'$ is \emph{more informative} than $\cN$
with respect to $\{\rho_s\}_s$ if there exists a family of CP-TP maps
$\{\Gamma_w\}_w$ such that
\begin{equation}
\Gamma_w \circ \cN'_w(\rho_s) = \cN_w(\rho_s),
\quad \forall\, (s,w ) \in \cS\times\cW.
\nn
\end{equation}

Furthermore, $\cN'$ is said to be strictly more informative than $\cN$ with respect to $\{\rho_s\}_s$  if $\cN'$ is more informative than $\cN$, but the converse does not hold.
\end{definition}

A point that is important in the present quantum setting is that a
learning algorithm is not modeled as a single quantum channel, but as a
quantum instrument.
Thus we write
\[
\mathcal N=\{\mathcal N_w\}_{w\in\mathcal W},
\qquad
\mathcal N'=\{\mathcal N'_w\}_{w\in\mathcal W},
\]
where the classical label \(w\) denotes the announced hypothesis and
\(\mathcal N_w\) is the corresponding completely positive map producing
the residual quantum output.
The normalization condition is that
\(\sum_{w\in\mathcal W}\mathcal N_w\) is trace preserving, and similarly
for \(\mathcal N'\).
In this instrument formulation, informativeness is compared branch by
branch, with respect to the same classical hypothesis \(w\).
We say that
an algorithm \(\mathcal N'\) is more informative than \(\mathcal N\) with
respect to the ensemble \(\{\rho_s\}_{s\in\mathcal S}\) if there exists a
post-processing CPTP map \(\Gamma\), called a simulation map, such that
\begin{equation}
\mathcal N_w(\rho_s)
=
\Gamma\!\left(\mathcal N'_w(\rho_s)\right),
\qquad
\forall s\in\mathcal S,\ \forall w\in\mathcal W .
\label{sim_eq}
\end{equation}
Here \(\Gamma\) acts only on the residual quantum output system.
It is
not an instrument and it does not generate, randomize, relabel, or
coarse-grain the classical hypothesis \(w\).
Rather, after the same
branch \(w\) has been selected, \(\Gamma\) simulates the residual quantum
output of \(\mathcal N\) from that of \(\mathcal N'\).

Thus, ITA should not be interpreted as a privacy condition by itself.
Rather, it is a credibility condition for privacy claims in the
untrusted setting.
It rules out algorithms whose apparent privacy is
produced only by artificial post-processing noise after a more
informative computation has already been performed.
Once ITA is
imposed, any remaining privacy must come from the physical
indistinguishability of the encoded states or from genuine limitations
of the allowed quantum operation.

If such a relation holds, the data-processing inequality
\cite{Petz1985Quasientropy} implies that the mutual information between
the input and the output of \(\mathcal N'\) is greater than or equal to
that of \(\mathcal N\).
Strict informativeness means that the above
simulation relation holds from \(\mathcal N'\) to \(\mathcal N\), while
\(\mathcal N'\) cannot itself be simulated from \(\mathcal N\) by a
CPTP map of the same branch-wise form (as mentioned in \eqref{sim_eq}).

To certify that a prescribed algorithm \(\mathcal N\) dominates over
every other algorithm \(\mathcal N'\) with respect to the collection
\(\{\rho_s\}_s\), we introduce the concept of Information-Theoretic
Admissibility (ITA).
\begin{definition}[Information-Theoretic Admissibility]\label{def:ITA}
A learning algorithm \(\mathcal N=\{\mathcal N_w\}_{w\in\mathcal W}\) is
ITA with respect to the set \(\{\rho_s\}_s\) if there exists no other
learning algorithm
\(\mathcal N'=\{\mathcal N'_w\}_{w\in\mathcal W}\) that is
\textit{strictly} more informative than \(\mathcal N\).
Equivalently,
there is no \(\mathcal N'\) such that \(\mathcal N\) can be obtained from
\(\mathcal N'\) through a simulation map \(\Gamma\) satisfying
\eqref{sim_eq}, while \(\mathcal N'\) cannot be obtained from
\(\mathcal N\) through a CPTP simulation map of the same branch-wise
form.
\end{definition}

Essentially, if an algorithm is ITA, it implies that the Data Processor
is already performing the optimal information extraction allowed by the
quantum mechanics formalism within the specified class of instruments
and for the specified ensemble.

The definition is ensemble-dependent.  It does not compare instruments
on all possible input states, but only on the particular family of
encoded states \(\{\rho_s\}_s\) supplied by the Respondent.
This is
necessary because privacy in the present model is a property of the
concrete data-encoding scheme, not of the instrument in isolation.
Consequently, an algorithm may be admissible for one ensemble and
inadmissible for another.
\subsection{Quantum Advantage: Privacy under ITA}

The imposition of the ITA condition reveals a fundamental divergence
between classical and quantum privacy capabilities.
In the classical
case, admissibility in the untrusted setting collapses privacy in a
strong sense: if the prescribed algorithm fails to retain enough
information to reconstruct the raw input, then a Data Processor can
instead keep a more informative classical representation and simulate
the prescribed algorithm afterwards.
In the quantum case, however,
non-commutativity changes this conclusion.  Since non-orthogonal
quantum states cannot be perfectly distinguished, an information-
theoretically admissible procedure need not permit perfect recovery of
the underlying classical data.

\noindent\textbf{The Collapse of Classical Privacy:}
In the classical domain, the encoded states \(\rho_s\) effectively
behave as probability distributions, or equivalently as mutually
commuting states.
In such a scenario, the lemma below shows that
simultaneous ITA, interpreted as maximality in the informativeness
order, and nontrivial privacy are incompatible unless privacy is already
imposed at the level of the raw states.
\begin{lemma}\label{lem:ITA_classical}
Assume that all states \(\{\rho_s\}_s\) commute, i.e., the setting is
classical.
If there exists no reconstruction map \(\Gamma\) such that
\[
\Gamma \circ \mathcal N(\rho_s)=\rho_s
\]
for every \(s\), then the algorithm \(\mathcal N\) is not ITA.
\end{lemma}

\begin{proof}
See Appendix \ref{proof_LL}.
\end{proof}

The intuition is straightforward.  When the encoded states commute, the
situation is effectively classical.
If the algorithm does not retain
enough information to reconstruct the input, then one can construct a
strictly more informative procedure by keeping the classical input label
and then simulating the original algorithm as a post-processing.
Therefore, a classical algorithm that is not recoverable cannot be
maximal in the information-theoretic order.
This is the sense in which
admissibility collapses privacy in the classical untrusted setting.

The implication of Lemma \ref{lem:ITA_classical} is severe: classical
ITA algorithms permit full reconstruction of the raw data.
Since under
ITA the Data Processor effectively holds the raw data, the output-based
guarantees of Definition \ref{DPL2} are insufficient.
Therefore, for
classical ITA algorithms, the privacy condition in Definition
\ref{DPL2} must be strengthened to require indistinguishability on the
raw states \(\rho_s\) directly, effectively substituting
\(\mathcal N(\rho_s)\) with \(\rho_s\).
We discuss this point in the further
in subsection \ref{subsec_ITA_adv}.

\noindent\textbf{Quantum Privacy via Non-Commutativity:}
In the quantum regime, the raw inputs \(\{\rho_s\}_s\) may be
non-commuting.
Quantum mechanics dictates that non-orthogonal states
cannot be perfectly distinguished, even by the optimal measurement.
Therefore, the classical equivalence between optimal information
extraction and perfect recovery no longer holds.
A quantum algorithm
may be maximal within the relevant instrument order, and hence ITA,
while still failing to reveal the underlying classical label \(s\)
perfectly.

This distinction is central to the present paper.  ITA is not a privacy
condition by itself.
Rather, it is a credibility condition for privacy
claims in the untrusted setting: it rules out apparent privacy that is
created only by applying artificial noise after a more informative
procedure has already been performed.
In the quantum case, once such
post-processing explanations are excluded, any remaining privacy may
come from the physical indistinguishability of the encoded states or
from genuine limitations of the allowed quantum operation.

The relevant comparison is the branch-wise instrument comparison
introduced above.  Thus an algorithm
\(\mathcal N=\{\mathcal N_w\}_{w\in\mathcal W}\) is compared with another
instrument \(\mathcal N'=\{\mathcal N'_w\}_{w\in\mathcal W}\) on the
specific ensemble \(\{\rho_s\}_s\).
A simulation map \(\Gamma\) is a
CPTP map acting only on the residual quantum output system and satisfies
\[
\mathcal N_w(\rho_s)
=
\Gamma\!\left(\mathcal N'_w(\rho_s)\right),
\qquad
\forall s,\ \forall w .
\]
It is not an instrument and does not generate, randomize, relabel, or
coarse-grain the classical hypothesis \(w\).
Hence the question of ITA
is not whether the prescribed algorithm is globally non-degradable as a
channel on arbitrary input states, but whether it is maximal on the
particular encoded ensemble under this branch-wise simulation order.
We now give a concrete quantum example in which ITA and privacy are
compatible.
\begin{example}[Quantum ITA algorithm with nontrivial privacy]
\label{ex:ITA_quantum}
Consider the states
\[
\rho_z=\ketbra{\phi_{z,p}},
\qquad
\ket{\phi_{z,p}}
=
\sqrt{1-p}\ket{0}+(-1)^z\sqrt{p}\ket{1}.
\]
For a dataset \(s=(z_1,\dots,z_n)\), the actual encoded state is
$\rho_s
=
\bigotimes_{j=1}^n \ketbra{\phi_{z_j,p}} $.
To define the measurement, we use the auxiliary orthonormal basis
obtained at \(p=1/2\),
$\ket{e_s}:=
\bigotimes_{j=1}^n \ket{\phi_{z_j,1/2}} $.
Let \(P_k\) be the projection onto the subspace spanned by
$\{\ket{e_s}: |s|_1=k\}$,
i.e., the subspace corresponding to strings with Hamming weight \(k\).
The untrusted Data Processor applies the projective instrument
\(\{\mathcal N_k\}_k\) defined by
$\mathcal N_k(\rho):=P_k\rho P_k $.
\end{example}

In Example \ref{ex:ITA_quantum}, the projective instrument
\(\{\mathcal N_k\}_k\) is ITA with respect to the ensemble
\(\{\rho_s\}_s\) under the branch-wise instrument order defined above. The point
of the example is that the ITA property and the privacy property arise
from different parts of the argument.
The ITA property is an
ensemble-dependent maximality statement in the instrument order, whereas
the privacy property comes from the non-orthogonality of the encoded
states.

\noindent\textbf{Justification for ITA:}
The projectors \(\{P_k\}_k\) define the Hamming-weight subspaces, and
the maps
\[
\mathcal N_k(\rho)=P_k\rho P_k,
\]
form a projective instrument.
The assertion that this instrument is ITA
is not inferred merely from the non-orthogonality of the states, nor
from an informal sufficient-statistic argument.
Rather, it is proved by
checking the ensemble-dependent branch-wise simulation order: there is
no instrument \(\mathcal N'=\{\mathcal N'_k\}_k\) that is strictly more
informative than \(\mathcal N\) while still satisfying
\[
\mathcal N_k(\rho_s)
=
\Gamma\!\left(\mathcal N'_k(\rho_s)\right),
\qquad
\forall s,\ \forall k ,
\]
for some CPTP map \(\Gamma\) acting only on the residual quantum output
system.
The map \(\Gamma\) is not an instrument and does not modify the
classical outcome \(k\).

In the orthogonal case \(p=1/2\), the states \(\{\rho_s\}_s\) are
mutually distinguishable in the basis \(\{\ket{e_s}\}_s\).
The
projectors \(\{P_k\}_k\) are then projections onto Hamming-weight
subspaces, and the corresponding operation is quantum non-demolition
with respect to this coarse-grained observable \cite{Braginsky1980}.
In this classical limit, however, the collapse phenomenon described
above applies: if the raw string can be perfectly recovered, then ITA
does not by itself provide a privacy advantage.

\noindent\textbf{Justification for Privacy:}
For \(p\neq 1/2\), the single-qubit states
\(\ket{\phi_{0,p}}\) and \(\ket{\phi_{1,p}}\) are non-orthogonal.
Consequently, the product states \(\{\rho_s\}_s\) are not perfectly
distinguishable in general.
Thus, even though the Hamming-weight
instrument is ITA in the branch-wise, the Investigator cannot perfectly identify the
underlying string \(s\).
The residual uncertainty is a consequence of
the intrinsic indistinguishability of the encoded quantum states, as
quantified operationally by state-discrimination bounds such as the
Helstrom bound \cite{Helstrom1967_2}.
This establishes the desired separation from the classical case.  In
the classical setting, ITA forces the Data Processor to hold enough
information to reconstruct the raw data.
In the quantum setting of
Example \ref{ex:ITA_quantum}, the Data Processor performs an
information-theoretically admissible procedure, but perfect recovery of
the classical string is still obstructed by the geometry of the
non-orthogonal encoding.
Hence the privacy is not produced by adding
noise after a non-private learning algorithm;
it arises from the
intrinsic quantum limitation on distinguishability during the learning
process itself.
In this sense, the security condition in Definition
\ref{DPL2} remains meaningful for the ITA algorithm in the example.

{ \subsection{Implications of ITA: Source-Layer Privacy and Quantum Intrinsic Noise}\label{subsec_ITA_adv}

~~~~~~~\textbf{Resolving the ITA Conflict: Source-Layer Privacy.}
The impossibility result in Lemma \ref{lem:ITA_classical} implies that in the classical domain, if a Data Processor is untrusted and executes an ITA (optimal) algorithm, privacy cannot be preserved by the algorithm itself.
Since an ITA algorithm extracts all available information, the output effectively reveals the raw input.
Consequently, to preserve privacy in the classical untrusted setting, the burden of protection must shift from the \textit{algorithm} to the \textit{input data} itself.
This is standardly achieved via \textit{Input Perturbation} or \textit{Local Differential Privacy} (LDP), where the Respondent applies a local randomization mechanism $\mathcal{M}$ to generate a noisy version $\tilde{s} = \mathcal{M}(s)$.
Even if the Untrusted Processor fully recovers $\tilde{s}$ (as allowed under ITA), the underlying sensitive data $s$ remains protected by the noise added at the source.
Thus, indistinguishability is enforced at the source layer, making the specific choice of the processor's algorithm irrelevant to the privacy guarantee.

\textbf{Quantum Encoding as Intrinsic Source Noise.}
This necessity for source-layer protection provides a rigorous motivation for our quantum learning framework.
In our model, the encoding map $s \mapsto \rho_s$ plays a role conceptually equivalent to classical input perturbation, but with a fundamental physical advantage.
In the classical setting, distinct data points $s \neq s'$ are perfectly distinguishable unless artificial noise is added.
In the quantum setting, however, if the encoded states $\{\rho_s\}_s$ are \textit{non-orthogonal}, they are physically indistinguishable with certainty.
This non-orthogonality introduces an intrinsic, unavoidable uncertainty—effectively ``quantum noise''—that prevents even an adversary with unlimited computational power from perfectly distinguishing $s$ from $s'$.
Therefore, our framework intrinsically embeds privacy into the physical layer.
Even if the Untrusted Data Processor employs an ITA algorithm (i.e., performs the optimal Helstrom measurement to extract maximum information), their ability to infer $s$ is fundamentally limited by the \textit{non-orthogonality} of the encoded states.
This confirms that our security condition is robust: privacy is not contingent on the Processor's cooperation but is guaranteed by the physical nature of the encoding itself.
}

\subsection{Distinction from Degradability}

Finally, it is crucial to distinguish ITA from the concept of
\textit{quantum channel degradation}
\cite{PhysRevA.85.012326,Buscemi2016}.
Degradability asks whether an
algorithm can be simulated for \textit{any arbitrary input state}.
In
contrast, ITA only asks whether the algorithm can be simulated on the
\textit{specific training ensemble} \(\{\rho_s\}_s\), and under the
branch-wise instrument order relevant to the announced classical
hypothesis.
An algorithm might be non-degradable, and hence secure in a
global channel-theoretic sense, but still simulable on the specific
subspace or family of states occupied by the Respondent's data.
Conversely, an ensemble-dependent ITA claim need not imply global
non-degradability.  Therefore, privacy certification in the untrusted
regime must be data-dependent, verifying admissibility explicitly
against the geometry of the Respondent's encoded states.
\section{Conclusion}
\label{sec:conclusion}
We established an information-theoretic framework for quantum generalization, demonstrating that limited information leakage controls expected generalization error (Theorem \ref{exp_gen_bound}).
Going beyond expected error bounds, we derived a bound on generalization error in probability (Theorem \ref{thm:renyi_gen_bound}) via Sandwiched R\'enyi divergence and a complementary lower bound on true loss (Theorem \ref{thm:true_loss_lb}), effectively sandwiching the risk under a newly introduced Classical-Quantum Sub-Gaussianity (Definition \ref{ass_cq_subg}).

We further established $(\eps, \delta)$-QDP as a sufficient condition for generalization by deriving a mechanism-agnostic stability bound on the Holevo information (Theorem \ref{lemma_Holevo_stability}) with logarithmic sample scaling by employing a grid-covering optimization to rigorously handle approximate privacy.
Finally, via Information-Theoretic Admissibility (ITA), we demonstrated a fundamental quantum advantage: unlike the classical regime, 
quantum mechanics permits admissible algorithms for which nontrivial
privacy guarantees remain meaningful against untrusted Data Processors.
\bibliographystyle{IEEEtran}
\bibliography{master}
\clearpage
\appendices

\section*{Organization of the Appendix}
The Appendix is organized into five thematic parts to support the main results:

\begin{itemize}
    \item \textbf{Generalization Error Bounds (Proofs):} 
    We provide the complete derivations for our generalization guarantees.
Appendix \ref{proof_exp_gen_bound} contains the proof of the expected generalization bound (Theorem \ref{exp_gen_bound}).
Appendix \ref{app:true_los_lb} derives the proof of lower bound on the expected true loss (Theorem \ref{thm:true_loss_lb}).
Appendix \ref{app:renyi_bound} establishes the proof of probabilistic upper-bound on generalization error via Sandwiched R\'enyi divergence (Theorem \ref{thm:renyi_gen_bound}) under the i.i.d. assumption.

    \item \textbf{Comparisons with Prior Work:} 
    We explicitly contrast our results with existing literature.
Appendix \ref{App:comp_gen} compares our generalization bounds with those of \cite{Caro23} and \cite{WDH2025}, including a detailed numerical analysis (Appendix \ref{app:numerical_comparison}).
Appendix \ref{App:comp_stab} contrasts our stability bounds with the results of \cite{Nuradha25} and \cite{Caro23}.
\item \textbf{Stability and Privacy Proofs:} 
    Appendix \ref{proof_lemma_Holevo_stability} provides the rigorous proof for the stability of $1$-neighbor $(\eps,\delta)$-DP algorithms (Theorem \ref{lemma_Holevo_stability}), along with the proofs for the pure DP and high-privacy regimes.
Appendix \ref{fact_QDP_group_proof} details the group privacy degradation properties.

\item \textbf{Technical Lemmas:} 
    Appendices \ref{proof_fact:Dmax_bound} and \ref{proof_lemma_mixture_quant_KL} contain proofs for auxiliary information-theoretic inequalities used throughout the stability analysis.
\end{itemize}
\section{Proof of Theorem \ref{exp_gen_bound}}\label{proof_exp_gen_bound}
Given the fact that $I[S\Te ;WB'] = D(\sigma^{S\Te WB}_{\cN}\|\sigma^{S\Te} \otimes \sigma^{WB'}_{\cN})$, we can lower-bound $D(\sigma^{S\Te WB}_{\cN}$ $\|\sigma^{S\Te} \otimes \sigma^{WB'}_{\cN})$ as follows,
    \begin{align}
        &~~~D(\sigma^{S\Te WB}_{\cN}\|\sigma^{S\Te} \otimes \sigma^{WB'}_{\cN})\nn\\ &\overset{(a)}{\geq} \tr\left[\lambda L^{S\Te WB}\sigma^{S\Te WB}_{\cN}\right] - \ln\tr\left[e^{\lambda L^{S\Te WB}}\left(\sigma^{S\Te} \otimes \sigma^{WB'}_{\cN}\right)\right]\nn\\
        &= \tr\left[\lambda L^{S\Te WB}\sigma^{S\Te WB}_{\cN}\right] -\tr\left[\lambda L^{S\Te WB}\left(\sigma^{S} \otimes \sigma^{WB'}_{\cN}\right)\right]\nn\\
        &\hspace{50pt} - \ln\tr\left[e^{\lambda \left(L^{S\Te WB} - \tr\left[\lambda L^{S\Te WB}\left(\sigma^{S\Te} \otimes \sigma^{WB'}_{\cN}\right)\right]\bbI^{S\Te WB}\right)}\left(\sigma^{S\Te} \otimes \sigma^{WB'}_{\cN}\right)\right]\nn\\
        &\overset{(b)}{\geq} \tr\left[\lambda L^{S\Te WB}\sigma^{S\Te WB}_{\cN}\right] 
- \tr\left[\lambda L^{S\Te WB}\left(\sigma^{S\Te} \otimes \sigma^{WB'}_{\cN}\right)\right] - \frac{\lambda^2 \alpha^2}{2},\label{proof_exp_gen_eq}
    \end{align}
where $(a)$ follows from the variational lower-bound for the quantum relative entropy (see \cite[Theorem $5.9$]{H2017QIT}) and $(b)$ follows from \eqref{ass_cq_subg_eq2} of Definition \ref{ass_cq_subg}.
Further, we can rewrite \eqref{proof_exp_gen_eq} as follows,
\begin{align}
    &\frac{\lambda^2 \alpha^2}{2} - \lambda\left(\tr\left[L^{S\Te WB}\sigma^{S\Te WB}_{\cN}\right] - \tr\left[ L^{S\Te WB}\left(\sigma^{S\Te} \otimes \sigma^{WB'}_{\cN}\right)\right]\right)\nn\\
    &\hspace{250pt}+ D(\sigma^{S\Te WB}_{\cN}\|\sigma^{S\Te} \otimes \sigma^{WB'}_{\cN}) \geq 0,\nn
\end{align}

Since the above inequality is a non-negative quadratic equation in $\lambda$ with the coefficient $\frac{\alpha^2}{2} \geq 0,$ its discriminant must be non-positive.
Hence, 
\begin{align}
    &\left(\tr\left[L^{S\Te WB}\sigma^{S\Te WB}_{\cN}\right] - \tr\left[ L^{S\Te WB}\left(\sigma^{S\Te} \otimes \sigma^{WB'}_{\cN}\right)\right]\right)^{2}\hspace{100pt}\nn\\
    &\leq 4\cdot\frac{\alpha^2}{2}\cdot D(\sigma^{S\Te WB}_{\cN}\|\sigma^{S\Te} \otimes \sigma^{WB'}_{\cN}).
\end{align}

Thus, we have, 
\begin{align}
\abs{\tr\left[L^{S\Te WB}\sigma^{S\Te WB}_{\cN}\right] - \tr\left[ L^{S\Te WB}\left(\sigma^{S\Te} \otimes \sigma^{WB'}_{\cN}\right)\right]} \leq \sqrt{2\alpha^2I[S\Te;WB']}.\label{proof_exp_gen_eq3}
\end{align}

Therefore, the combination of  Definition \ref{gen_ws_error_exp} and \eqref{proof_exp_gen_eq3} yields the following,
\begin{align}
        \overline{\text{\textnormal{gen}}}_{\rho}(\cN) &\leq \sqrt{2\alpha^2 I[S\Te;WB']}.\label{proof_exp_gen_eq4}
    \end{align}

 This completes the proof of Theorem \ref{exp_gen_bound}.\hfill\QED

{

\section{Proof of Theorem \ref{thm:renyi_gen_bound}}\label{app:renyi_bound}

In this section, we establish an upper-bound on the absolute generalization error in probability using the Sandwiched R\'enyi divergence. In contrast to the bound on the expected generalization error (Theorem \ref{exp_gen_bound}), which is based on standard Mutual Information 
and only controls the error on average, our focus here is on obtaining guarantees that hold with high probability (confidence level $1-\delta$). Such guarantees are essential in safety-critical settings, where average-case performance is not enough;
one must ensure the error remains small with high confidence.

Achieving this type of high-probability guarantee requires controlling higher-order moments of the dependence between the data and the learned hypothesis.
This dependence is quantified by the Sandwiched R\'enyi divergence. Even though the sandwiched R\'enyi divergence ($\Tilde{D}_{\gamma}$) gives a larger value than standard Mutual Information for $\gamma \in (1, \infty)$, using this stronger measure allows us to separate the randomness of the loss function from how much the algorithm relies on the training data with the help of non-commutative H\"older's inequality \cite{bhatia2013matrix}. This separation is necessary to guarantee that the model performs well even in the worst-case scenarios.

Before delving into the formal proof, it is crucial to establish how the independent and identically distributed (i.i.d.) structure of the dataset influences the tail behavior of the global loss operator. The theorem relies on the local sub-Gaussianity of individual data points, which naturally translates to a tighter concentration for the global aggregate.

\begin{remark}[Scaling of the Global Variance]
It is important to highlight that the local assumption \eqref{ass_cq_subg_loc_cond_eq} implies a strictly tighter bound for the global loss operator.
Specifically, due to the tensor product structure and the independence of $Z_i$, the global condition holds with a variance proxy that scales as $1/n$,
    \begin{equation}
        \tr\left[e^{\lambda\left(L^{S\Te B}_{w} - \tr\left[L^{S\Te B'}_{w}(\sigma^{S \Te} \otimes (\sigma^{\cN}_{w})^{B'}) \right]\bbI^{S\Te B}\right)}(\sigma^{S \Te} \otimes (\sigma^{\cN}_{w})^{B'})\right] \leq e^{\frac{\lambda^2 \alpha^2}{2n}}.\label{ass_cq_subg_loc_cond_eq2}
    \end{equation}
where, for every $w \in \cW$, we define, $L^{S\Te B'}_{w} :=  \sum_{s \in \cS}\ketbrasys{s}{S} \otimes \overset{\Te B'}{L(s,w)}.$ This scaling is a direct consequence of the additivity of cumulants for independent variables (or technically, via Jensen's inequality for the operator exponential).
This $1/n$ factor is precisely what allows the generalization bound to vanish as the dataset size increases.
\end{remark}

 A direct consequence of this scaling is that the expected generalization error bound from Theorem \ref{exp_gen_bound} also benefits from the sample size.
Substituting the variance proxy $\frac{\alpha^2}{n}$ into the framework of Theorem \ref{exp_gen_bound} yields us the following result

\begin{corollary}[Expected Generalization Bound under I.I.D.
Assumption]
\label{corr:iid_expected_bound}
 For a fixed $\alpha \in (0,\infty),$ if the loss operators for a quantum learning algorithm $\cN$, satisfy Definition \ref{ass_cq_subg_loc}, then, we have,
\begin{equation}
\overline{\text{\textnormal{gen}}}_{\rho}(\cN) \leq \sqrt{\frac{2\alpha^2}{n} I[S \Te;WB']}.\nn
\end{equation}
\end{corollary}

This explicitly demonstrates the $\sqrt{1/n}$ convergence rate for the expected error, confirming that algorithmic stability (bounded mutual information) leads to vanishing generalization error as $n \to \infty$.

We now prove Theorem \ref{thm:renyi_gen_bound} by defining  $\cE_{S,W} := \tr\left[L(S,W)(\sigma^{\cN}_{S,W})^{\Te B'}\right] - L_{\rho} (\cN,W)$.
We are interested in bounding $\Pr\{|\cE_{S,W}| > \eps\}$. By the union bound, this probability is written as,
\begin{align}
    \Pr\{|\cE_{S,W}|
> \eps\} = \Pr\{\cE_{S,W} > \eps\} + \Pr\{\cE_{S,W} < -\eps\} = \Pr\{\cE_{S,W} > \eps\} + \Pr\{-\cE_{S,W} > \eps\}.\nn
\end{align}
We first bound the positive deviation $\Pr\{\cE_{S,W} > \eps\}$.
Applying the Markov inequality for any $\lambda > 0$, we have,

\begin{align} &~~~\Pr\{\cE_{S,W} > \eps\}\nn\\
&= \Pr_{(S,W) \sim P^{\cN}_{SW}}\left\{\tr\left[(L(S,W) - L_{\rho} (\cN,W) \bbI)(\sigma^{\cN}_{S,W})^{\Te B'}\right] > \eps\right\}\nn\\ 
&= E_{W \sim P^{\cN}_{W}} \left[\Pr_{S \sim P^{\cN}_{S|W}}\left\{\tr\left[(L(S,W) - L_{\rho} (\cN,W) \bbI)(\sigma^{\cN}_{S,W})^{\Te B'}\right] > \eps\right\}\right].
\label{eq:initial} \end{align}

For a fixed $w \in \cW$, we have,

\begin{align} &~~~\Pr_{S \sim P^{\cN}_{S|W=w}}\left\{\tr\left[(L(S,w) - L_{\rho} (\cN,w) \bbI)(\sigma^{\cN}_{S,w})^{\Te B'}\right] > \eps\right\}\nn\\ 
&\overset{(a)}{\leq} \Pr_{S \sim P^{\cN}_{S|W=w}}\left\{\tr\left[e^{\lambda(L(S,w) - L_{\rho} (\cN,w) \bbI)}(\sigma^{\cN}_{S,w})^{\Te B'}\right] > e^{\lambda \eps}\right\} \nn\\
&\leq e^{-\lambda\eps} \bbE_{S|W=w} \left[\tr\left[e^{\lambda (L(S,w) - L_{\rho} (\cN,w)\bbI)} (\sigma^{\cN}_{S,w})^{\Te B'}\right]\right]\nn\\
&= e^{-\lambda\eps} \tr\left[e^{\lambda(L^{S\Te B'}_{w}- L_{\rho} (\cN,w)\bbI)}\sigma^{S\Te B'}_{\cN,w}\right].
\label{eq:chernoff_initial} \end{align}
where ($a$) follows from the fact that for any Hermitian operator $H$ and density matrix $\rho$, the convexity bound $\tr[e^H \rho] \geq e^{\tr[H \rho]}$ holds and in \eqref{eq:chernoff_initial}, we define $\sigma^{S\Te B'}_{\cN,w} := \sum_{s \in \cS} P^{\cN}_{S|W=w}(s)$ $\ketbra{s} \otimes  (\sigma^{\cN}_{s,w})^{\Te B'}$.
To bound the trace term in Eq. \eqref{eq:chernoff_initial}, we introduce the product state $\sigma_{\text{prod},w} := \sigma^{S\Te} \otimes(\sigma^{\cN}_{w})^{B'} $.
Thus, we have,
\begin{align}
    &~~~\tr\left[e^{\lambda(L^{S\Te B'}_{w}- L_{\rho} (\cN,w)\bbI)}\sigma^{S\Te B'}_{\cN,w}\right]\nn\\ 
    &= \tr\left[ \left(\sigma_{\text{prod},w}^{\frac{\gamma-1}{2\gamma}} e^{\lambda(L^{S\Te B'}_{w} - L_{\rho} (\cN,w)\bbI)} \sigma_{\text{prod},w}^{\frac{\gamma-1}{2\gamma}}\right) \left(\sigma_{\text{prod},w}^{\frac{1-\gamma}{2\gamma}} \sigma^{S\Te B'}_{\cN,w} \sigma_{\text{prod},w}^{\frac{1-\gamma}{2\gamma}}\right) \right].\label{eq:term_comp}
\end{align}
We invoke the non-commutative H\"older's inequality,
$$ \abs{\tr[AB]} \leq \left( \tr[\abs{A}^p]\right)^{\frac{1}{p}}\left( \tr[\abs{B}^q]\right)^{\frac{1}{q}},$$
choosing the conjugate exponents $q = \gamma$ and $p = \frac{\gamma}{\gamma-1}$, and defining the operators
\[
A := \sigma_{\text{prod},w}^{\frac{\gamma-1}{2\gamma}} e^{\lambda(L^{S\Te B'}_{w} - L_{\rho} (\cN,w)\bbI)} \sigma_{\text{prod},w}^{\frac{\gamma-1}{2\gamma}}, 
\quad\quad 
B := \sigma_{\text{prod},w}^{\frac{1-\gamma}{2\gamma}} \sigma^{S\Te B'}_{\cN,w} \sigma_{\text{prod},w}^{\frac{1-\gamma}{2\gamma}}.
\]
With this choice, the inequality specializes to $\tr[AB] \leq \underbrace{\left( \tr[\abs{A}^p]\right)^{\frac{1}{p}}}_{\textbf{Term I}}\underbrace{\left( \tr[\abs{B}^\gamma]\right)^{\frac{1}{\gamma}}}_{\textbf{Term II}}$.

\paragraph{1.
Analysis of Term I (Algorithm's Data Dependency):}
By the definition of sandwiched R\'enyi divergence and the fact that $B$ is a positive operator, we have,
\begin{align}
   \left( \tr[\abs{B}^\gamma]\right)^{\frac{1}{\gamma}} &= \left( \tr\left[ \left( \sigma_{\text{prod},w}^{\frac{1-\gamma}{2\gamma}} \sigma^{S\Te B'}_{\cN,w} \sigma_{\text{prod},w}^{\frac{1-\gamma}{2\gamma}} \right)^\gamma \right] \right)^{\frac{1}{\gamma}} 
    = \exp\left( \frac{\gamma-1}{\gamma} \Tilde{D}_{\gamma} (\sigma^{S\Te B'}_{\cN,w}\| \sigma_{\text{prod},w}) \right).
\label{eq:bound_B_full}
\end{align}

\paragraph{2. Analysis of Term II (Randomness of the Loss operators):}
The combination of the fact that $B$ is a positive operator and the Araki-Lieb-Thirring inequality ($\tr[(YXY)^r] \leq \tr[Y^{r}X^rY^{r}]$) with $r=p=\frac{\gamma}{\gamma-1}$, yields the following,
\begin{align}
    \left( \tr[\abs{A}^p]\right)^{\frac{1}{p}} &= \left( \tr\left[ \left( \sigma_{\text{prod},w}^{\frac{\gamma-1}{2\gamma}} e^{\lambda(L^{S\Te B'}_{w} - L_{\rho} (\cN,w)\bbI)} \sigma_{\text{prod},w}^{\frac{\gamma-1}{2\gamma}} \right)^p \right] \right)^{\frac{1}{p}} \nn\\
    &\leq \left( \tr\left[ \sigma^{1/2}_{\text{prod}} e^{\frac{\gamma\lambda}{\gamma-1}(L^{S\Te B'}_{w} - L_{\rho} (\cN,w)\bbI)}\sigma^{1/2}_{\text{prod}} \right] \right)^{\frac{\gamma-1}{\gamma}}\nn\\
    &= \left( \tr\left[  e^{\frac{\gamma\lambda}{\gamma-1}(L^{S\Te B'}_{w} - L_{\rho} (\cN,w)\bbI)}\sigma_{\text{prod},w} \right] \right)^{\frac{\gamma-1}{\gamma}}.
\end{align}

Invoking the Classical-Quantum Sub-Gaussian assumption (\eqref{ass_cq_subg_loc_cond_eq2}) for $\alpha$ and setting $\lambda \leftarrow  \frac{\gamma\lambda}{\gamma-1}$ in \eqref{ass_cq_subg_eq2}, we have,
\begin{align}
    \left( \tr[\abs{A}^p]\right)^{\frac{1}{p}} &\leq \left( \exp\left( \frac{1}{2n} \left(\frac{\gamma\lambda}{\gamma-1}\right)^2 \alpha^2 \right) \right)^{\frac{\gamma-1}{\gamma}} 
    = \exp\left( \frac{\gamma \lambda^2 \alpha^2}{2n(\gamma-1)} \right).
\label{eq:bound_A_full}
\end{align}

\paragraph{3. Aggregation and Global Divergence:}
Substituting \eqref{eq:bound_B_full} and \eqref{eq:bound_A_full} back into \eqref{eq:chernoff_initial}, and then computing the expectation over $W$ in \eqref{eq:initial}, yields us,
{\allowdisplaybreaks\begin{align}
&~~~\Pr\{\cE_{S,W} > \eps\}\nn\\
&\leq e^{-\lambda\eps} \exp\left( \frac{\gamma \lambda^2 \alpha^2}{2n(\gamma-1)} \right) \bbE_{W \sim P^{\cN}_{W}} \left[ \exp\left( \frac{\gamma-1}{\gamma} \Tilde{D}{\gamma} (\sigma^{S\Te B'}_{\cN,w}\|  \sigma_{\text{prod},w}) \right) \right]\nn\\
&\overset{(a)}{\leq} e^{-\lambda\eps} \exp\left( \frac{\gamma \lambda^2 \alpha^2}{2n(\gamma-1)} \right) \left(\bbE_{W \sim P^{\cN}_{W}} \left[ \exp\left((\gamma-1) \Tilde{D}{\gamma} (\sigma^{S\Te B'}_{\cN,w}\|  \sigma_{\text{prod},w}) \right) \right]\right)^{\frac{1}{\gamma}}\nn\\
&\overset{(b)}{\leq}  e^{-\lambda\eps} \exp\left( \frac{\gamma \lambda^2 \alpha^2}{2n(\gamma-1)} \right) \left(\exp\left((\gamma-1) \Tilde{D}{\gamma} (\sigma^{S\Te WB'}_{\cN}| \sigma^{S\Te} \otimes \sigma^{WB'}_{\cN}) \right) \right)^{\frac{1}{\gamma}}\nn\\
&=  \exp\left( -\lambda\eps +\frac{\gamma \lambda^2 \alpha^2}{2n(\gamma-1)} \right) \exp\left(\frac{\gamma-1}{\gamma} \Tilde{D}{\gamma} (\sigma^{S\Te WB'}_{\cN}|
\sigma^{S\Te} \otimes \sigma^{WB'}_{\cN}) \right) \nn\\
&= \exp\left( -\lambda\eps +\frac{\gamma \lambda^2 \alpha^2}{2n(\gamma-1)} + \frac{\gamma-1}{\gamma} \Tilde{D}{\gamma} (\sigma^{S\Te WB'}_{\cN}| \sigma^{S\Te} \otimes \sigma^{WB'}_{\cN}) \right)\label{final_bound_eq}
\end{align}}
where $(a)$ follows from Jensen's inequality and the concavity of the function $f(x) = x^{\frac{1}{\gamma}}$ for $\gamma > 1$ and $(b)$ follows since the joint state $\sigma^{S\Te WB'}_{\cN}$ and the product state $\sigma^{S\Te} \otimes \sigma^{WB'}_{\cN}$ are block-diagonal with respect to the classical system $W$, the divergence decomposes as,
\begin{equation}
\Tilde{D}{\gamma} (\sigma^{S\Te WB'}_{\cN}| \sigma^{S\Te} \otimes \sigma^{WB'}_{\cN}) = \frac{1}{\gamma-1} \ln \bbE_{W} \left[ \exp\left( (\gamma-1) \Tilde{D}{\gamma} (\sigma^{S\Te B'}_{\cN,w}\|  \sigma_{\text{prod},w}) \right) \right].\nn
\end{equation}

Minimizing the exponent $f(\lambda) = -\lambda\eps + \frac{\gamma \alpha^2}{2n(\gamma-1)} \lambda^2$ yields $\lambda^* = \frac{n\eps(\gamma-1)}{\gamma \alpha^2}$, resulting 
in,
\begin{equation}
    \Pr\{\cE_{S,W} > \eps\} \leq \exp\left( - \frac{\gamma-1}{\gamma} \left( \frac{n\eps^2}{2\alpha^2} - \Tilde{D}_{\gamma} (\sigma^{S\Te WB'}_{\cN}\|
\sigma^{S\Te} \otimes \sigma^{WB'}_{\cN}) \right) \right).
\end{equation}

By symmetry of the sub-Gaussian assumption, the same bound holds for the negative deviation $\Pr\{-\cE_{S,W} > \eps\}$.
Thus, for the absolute deviation $|Z|$, we have,
\begin{equation}
    \Pr\{|\cE_{S,W}|
> \eps\} \leq 2 \exp\left( - \frac{\gamma-1}{\gamma} \left( \frac{n\eps^2}{2\alpha^2} - \Tilde{D}_{\gamma} (\sigma^{S\Te WB'}_{\cN}\| \sigma^{S\Te} \otimes \sigma^{WB'}_{\cN}) \right) \right).
\end{equation}

\paragraph{4.
Inversion for High-Probability Guarantee:}
We set the upper bound equal to the confidence level $\delta \in (0,1)$,
\[
\delta = 2 \exp\left( - \frac{\gamma-1}{\gamma} \left( \frac{n\eps^2}{2\alpha^2} - \Tilde{D}_{\gamma} (\sigma^{S\Te WB'}_{\cN}\| \sigma^{S\Te} \otimes \sigma^{WB'}_{\cN}) \right) \right).
\]

Finally, solving for $\eps$, we have,
\[
\eps = \sqrt{\frac{2\alpha^2}{n} \left( \Tilde{D}_{\gamma} (\sigma^{S\Te WB'}_{\cN}\| \sigma^{S\Te} \otimes \sigma^{WB'}_{\cN}) + \frac{\gamma}{\gamma-1} \ln \frac{2}{\delta} \right)}.
\]
This completes the proof of Theorem \ref{thm:renyi_gen_bound}.\hfill\QED

\begin{remark}[Comparison with Classical Bounds]
\label{rem:esposito_recovery}
In the special case where the quantum input subsystem $\Te$ and the algorithm's internal quantum output $B'$ are trivial (i.e., the systems are purely classical), the term $\Tilde{D}_{\gamma} (\sigma^{S\Te WB'}_{\cN}\| \sigma^{S\Te} \otimes \sigma^{WB'}_{\cN})$ reduces to Sibson's mutual information \cite{Verdu15} $I_{\gamma}(S;W)$ of order $\gamma$.
Consequently, the bound derived in Theorem \ref{thm:renyi_gen_bound} recovers the exact form of Corollary 2 in \cite{Esposito21}.
\end{remark}

}

\section{Proof of Theorem \ref{thm:true_loss_lb}}\label{app:true_los_lb}

In this section, we establish a relationship between the expected true loss $L_{\rho}(\cN)$ (see Definition \ref{exp_true_loss_def}) and the expected empirical loss $\hat{L}_{\rho} (\cN)$ (see Definition \ref{exp_emp_loss_def}).
We denote the product state $\sigma_{\text{prod}} := \sigma^{S\Te} \otimes \sigma^{WB'}_{\cN}$.
Our goal is to upper bound the empirical loss $\hat{L}_{\rho}(\cN) = \tr[L^{S\Te WB'}\sigma^{S\Te WB'}_{\cN}]$ in terms of the true loss and the divergence.
We begin by expanding the trace using the identity $\mathbb{I} = \sigma_{\text{prod}}^{\frac{\gamma-1}{2\gamma}} \sigma_{\text{prod}}^{\frac{1-\gamma}{2\gamma}}$ and applying a series of information-theoretic inequalities.
Consider the following derivation,
\begin{align}
     \hat{L}_{\rho} (\cN) & = \tr[L^{S\Te WB'}\sigma^{S\Te WB'}_{\cN}]\nn\\
     &=  \tr\left[\left(\sigma_{\text{prod}}^{\frac{\gamma-1}{2\gamma}}L^{S\Te WB'}\sigma_{\text{prod}}^{\frac{\gamma-1}{2\gamma}}\right) \left(\sigma_{\text{prod}}^{\frac{1-\gamma}{2\gamma}}\sigma^{S\Te WB'}_{\cN}\sigma_{\text{prod}}^{\frac{1-\gamma}{2\gamma}}\right)\right]\nn\\
     &\overset{(a)}{\leq} \tr\left[\left|\sigma_{\text{prod}}^{\frac{\gamma-1}{2\gamma}}L^{S\Te WB'}\sigma_{\text{prod}}^{\frac{\gamma-1}{2\gamma}}\right|^{\frac{\gamma}{\gamma-1}}\right]^{\frac{\gamma-1}{\gamma}} \tr\left[\left|\sigma_{\text{prod}}^{\frac{1-\gamma}{2\gamma}}\sigma^{S\Te WB'}_{\cN}\sigma_{\text{prod}}^{\frac{1-\gamma}{2\gamma}}\right|^{\gamma}\right]^{\frac{1}{\gamma}}\nn\\
     &\overset{(b)}{\leq} \tr\left[\left(L^{S\Te WB'}\right)^{\frac{\gamma}{\gamma-1}}\sigma_{\text{prod}}\right]^{\frac{\gamma-1}{\gamma}}\exp\left( \frac{\gamma-1}{\gamma} \Tilde{D}_{\gamma} (\sigma^{S\Te WB'}_{\cN}\|
\sigma_{\text{prod}}) \right)\nn\\
     &\overset{(c)}{\leq} \tr\left[e^{\frac{\gamma}{\gamma-1}L^{S\Te WB'}}\sigma_{\text{prod}}\right]^{\frac{\gamma-1}{\gamma}}\exp\left( \frac{\gamma-1}{\gamma} \Tilde{D}_{\gamma} (\sigma^{S\Te WB'}_{\cN}\| \sigma_{\text{prod}}) \right)\nn\\
     &= \tr\left[e^{\frac{\gamma}{\gamma-1}\left(L^{S\Te WB'}- L_{\rho} (\cN) \bbI + L_{\rho} (\cN)\bbI\right)}\sigma_{\text{prod}}\right]^{\frac{\gamma-1}{\gamma}}\exp\left( \frac{\gamma-1}{\gamma} \Tilde{D}_{\gamma} (\sigma^{S\Te WB'}_{\cN}\| \sigma_{\text{prod}}) \right)\nn\\
     &= \tr\left[e^{\frac{\gamma}{\gamma-1}\left(L^{S\Te WB'}- L_{\rho} (\cN) \bbI\right)} e^{\frac{\gamma}{\gamma-1}L_{\rho} (\cN)\bbI}\sigma_{\text{prod}}\right]^{\frac{\gamma-1}{\gamma}}\exp\left( \frac{\gamma-1}{\gamma} \Tilde{D}_{\gamma} (\sigma^{S\Te WB'}_{\cN}\| \sigma_{\text{prod}}) \right)\nn\\
     &= \tr\left[e^{\frac{\gamma}{\gamma-1}\left(L^{S\Te WB'}- L_{\rho} (\cN) \bbI\right)} \sigma_{\text{prod}}\right]^{\frac{\gamma-1}{\gamma}}\exp\left(L_{\rho} (\cN) + \frac{\gamma-1}{\gamma} \Tilde{D}_{\gamma} (\sigma^{S\Te WB'}_{\cN}\|
\sigma_{\text{prod}}) \right)\nn\\
     &\overset{(d)}{\leq} \exp\left(\frac{\gamma \alpha^2}{2(\gamma-1)} + L_{\rho} (\cN) + \frac{\gamma-1}{\gamma} \Tilde{D}_{\gamma} (\sigma^{S\Te WB'}_{\cN}\| \sigma_{\text{prod}}) \right).\nn
\end{align}

\textbf{Justification of inequalities:}
\begin{itemize}
    \item[$(a)$] Follows from the non-commutative H\"older inequality \cite{bhatia2013matrix},
    \begin{equation}
        \tr[|AB|] \leq \left(\tr[A^p]\right)^{1/p} \left(\tr[B^q]\right)^{1/q},
    \end{equation} for any two positive opertors $A$ and $B$,
    where we identify the operators $A = \sigma_{\text{prod}}^{\frac{\gamma-1}{2\gamma}}L^{S\Te WB'}\sigma_{\text{prod}}^{\frac{\gamma-1}{2\gamma}}$ and $B = \sigma_{\text{prod}}^{\frac{1-\gamma}{2\gamma}}\sigma^{S\Te WB'}_{\cN}\sigma_{\text{prod}}^{\frac{1-\gamma}{2\gamma}}$ to be positive operators, with conjugate exponents $p = \frac{\gamma}{\gamma-1}$ and $q = \gamma$.

    \item[$(b)$] The first factor follows from the 
Araki-Lieb-Thirring inequality \cite{Araki1990,Lieb_Thirring_2005} $\tr[(B A B)^r] \leq \tr[A^rB^{2r} ]$ (with $B = \sigma_{\text{prod}}^{\frac{\gamma-1}{2\gamma}}$, $A = L^{S\Te WB'}$, and $r = \frac{\gamma}{\gamma-1}$).
The second factor follows directly from the definition of the Sandwiched R\'enyi divergence $\Tilde{D}_{\gamma}$.
\item[$(c)$] Follows from the operator inequality $X^p \leq e^{pX}$ for any positive operator $X \geq 0$ and $p > 0$.
Here, we apply this to the operator $X = L^{S\Te WB'}$ with $p = \frac{\gamma}{\gamma-1}$.
\item[$(d)$] Follows from the Classical-Quantum Sub-Gaussian assumption (Definition \ref{ass_cq_subg}). By setting $\lambda = \frac{\gamma}{\gamma-1}$, the assumption guarantees $\tr\left[e^{\lambda\left(L^{S\Te WB'}- L_{\rho} (\cN) \bbI\right)} \sigma_{\text{prod}}\right] \leq e^{\frac{\lambda^2 \alpha^2}{2}}$.
Raising this to the power of $\frac{1}{\lambda} = \frac{\gamma-1}{\gamma}$ yields the term $\exp\left(\frac{\lambda \alpha^2}{2}\right) = \exp\left(\frac{\gamma \alpha^2}{2(\gamma-1)}\right)$.
\end{itemize}

Rearranging the terms in the final inequality to lower-bound $\exp(L_{\rho}(\cN))$ yields the statement of the theorem.\hfill\QED

Further, if the loss observables $\{L(w,s)\}$ are strictly bounded between $0$ and $\bbI$, we can derive a tighter multiplicative lower bound on the expected true loss that does not depend on the sub-Gaussian parameter $\alpha$.
\begin{corollary}
\label{corr:true_loss_lb_bounded}
Let $\cN$ be a quantum learning algorithm. Assume the loss operators are bounded such that $0 \preceq L(w,s) \preceq \bbI$, for all $w,s$.
For  any sandwiched R\'enyi divergence order $\gamma > 1$, the expected true loss is lower bounded by the empirical loss as follows,
\begin{equation}
    L_{\rho} (\cN) \geq \hat{L}^{\frac{\gamma}{\gamma-1}}_{\rho}(\cN) \exp\left( - \Tilde{D}_{\gamma} (\sigma^{S\Te WB'}_{\cN}\| \sigma^{S\Te} \otimes \sigma^{WB'}_{\cN}) \right).
\end{equation}
\end{corollary}

\begin{remark}
\label{rem:esposito_analogue}
Corollary \ref{corr:true_loss_lb_bounded} can be viewed as a quantum analogue of \cite[Theorem 3]{Modak21}.
In the classical setting, \cite[Theorem 3]{Modak21} derived a similar lower bound on the true risk in terms of the empirical risk and the classical R\'enyi divergence under the assumption of bounded loss functions.
Our result extends this bound to the quantum learning framework, where the non-commutativity of the state and loss operators necessitates the use of the Sandwiched R\'enyi divergence $\Tilde{D}_{\gamma}$ (which reduces to the classical R\'enyi divergence when states commute) and the utilization of the non-commutative H\"older inequality to separate the statistical fluctuations from the dependency structure.
\end{remark}

\section{Comparison of various Generalization Error Bounds obtained in this manuscript with the Prior Works}\label{App:comp_gen}

In this section, we compare the upper-bounds on the generalization error (Theorem \ref{exp_gen_bound} and Theorem \ref{thm:renyi_gen_bound}) derived in this work—both in expectation and in probability—with existing results.
\subsection{Comparision of Theorem \ref{exp_gen_bound} with Corollary $23$ of \cite{Caro23} }\label{comp:caro}

We contrast our upper-bound on the expected generalization error (Theorem \ref{exp_gen_bound}) with Corollary $23$ of \cite{Caro23},
highlighting the key advantages of our framework.
\begin{itemize}
    \item \textbf{Simplified Sub-Gaussianity Assumptions:} In \cite[Corollary $23$]{Caro23}, the authors impose the following two separate point-wise holding sub-Gaussianity requirements,
\begin{small}
\begin{align}
\mathrm{Tr} \left[  e^{\lambda ( L(s, w) - f(s,w) \mathbb{I}^{\Te B'} )} (\rho^{\Te}_s \otimes (\sigma^{\mathcal{N}}_{s, w})^{B'}) \right] &\leq e^{ \frac{\mu^2\lambda^2}{2}}, \quad \forall (s,w) \in \cS\times\cW,\label{qmgf}\tag{QMGF}\\
\mathbb{E}_{S \sim P^m} \left[ e^{\lambda ( f(S, w) - \mathbb{E}_{\tilde{S}}[f(\tilde{S}, w)] )} \right] &\leq e^{ \frac{\tau^2\lambda^2}{2}}, \quad \forall w\in \cW.\label{cmgf}\tag{CMGF}
\end{align}
\end{small}
for some fixed $\mu,\tau >0$, where $f(s,w) \coloneqq \mathrm{Tr}[L(s, w)(\rho^{\Te}_s \otimes (\sigma^{\mathcal{N}}_{s, w})^{B'})]$.
However, unlike Eqs. \eqref{qmgf} and \eqref{cmgf}  in \cite{Caro23}, Theorem \ref{exp_gen_bound} requires only a single condition mentioned in Definition \ref{ass_cq_subg}.
Crucially, our assumption holds in expectation over $P_S \times P_W$, rather than for worst-case pairs.
\item \textbf{Unified Information Measure:}
The upper-bound obtained in \cite[Corollary $23$]{Caro23} contains separated classical and quantum information terms because of the separated sub-Gaussianity assumption mentioned in \eqref{qmgf} and \eqref{cmgf} respectively.
In constrast, our bound on $\overline{\text{\textnormal{gen}}}_{\rho}(\cN)$ in Theorem \ref{exp_gen_bound} relies on a single information-theoretic quantity, which unifies classical and quantum dependencies.
\item \textbf{Failure of Stability Implications:} In scenarios where testing and training data are uncorrelated, the bound in \cite[Corollary $23$]{Caro23} reduces to a purely classical term and therefore it will not account for the quantum system $B'$.
Therefore, from \cite[Corollary $23$]{Caro23} it is not possible to show that stability implies generalizability.
In contrast, Theorem \ref{exp_gen_bound} avoids this limitation, validating the definition of expected true loss proposed in \cite{WDH2025} as its correct formulation.
A justification for the correctness of \eqref{exp_true_loss_def} is also given in \cite{WDH2025}.
\end{itemize}

With this correct definition of true loss, the bound obtained in \cite[Corollary 23]{Caro23} translates to \cite[Theorem 1]{WDH2025}.
In Appendix \ref{app:numerical_comparison}, we make a comparison of Theorem \ref{exp_gen_bound} with \cite[Theorem 1]{WDH2025} for the case when $\alpha = \mu=\tau,$ where 
$ \alpha,\mu$ and $\tau$ 
denote the sub-Gaussianity parameters appearing in \eqref{ass_cq_subg_eq}, \eqref{qmgf}, and \eqref{cmgf}, respectively.
\subsection{Numerical Comparison of Theorem \ref{exp_gen_bound} with \cite[Theorem 1]{WDH2025}}
\label{app:numerical_comparison}

In this appendix, we numerically validate our theoretical results by comparing them against the bounds established in \cite{WDH2025}.
We utilize the classical-quantum toy example described in \cite[Section \RNum{6}]{WDH2025} to demonstrate the tightness of our mutual information-based approach.
For this comparison, we evaluate the following two quantities under the condition that the sub-Gaussianity parameters satisfy $\mu=\tau=\alpha$.
\begin{figure}[htbp]
  \centering
  \begin{subfigure}[b]{0.9\textwidth}
  \centering
    \includegraphics[width=0.75\textwidth]{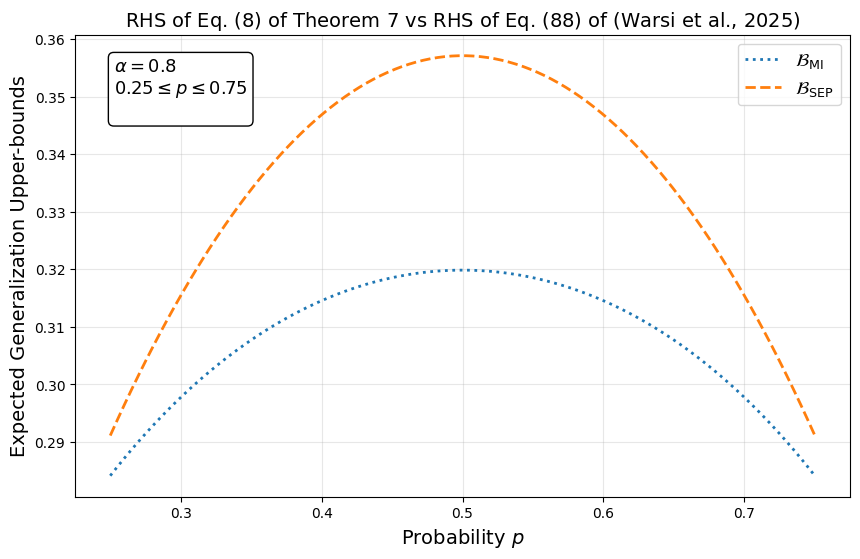}
    \label{fig:var_p}
  \end{subfigure}
  \hfill\\
  \begin{subfigure}[b]{0.9\textwidth}
  \centering
    \includegraphics[width=0.75\textwidth]{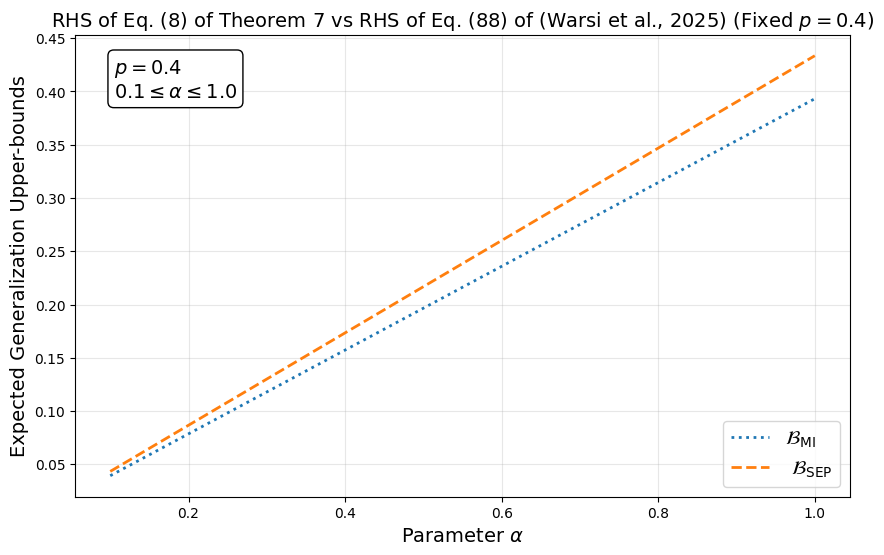} 
    \label{fig:var_tau}
  \end{subfigure}
    \caption{Numerical comparison of the generalization error bounds for the classical-quantum toy example in \cite{WDH2025}. 
    \textbf{(a)} Comparison of $\mathcal{B}_{\text{MI}}$ \eqref{B_IT} and $\mathcal{B}_{\text{SEP}}$ \eqref{B_SEP} as a function of the prior probability $p \in [0.25, 0.75]$. 
    \textbf{(b)} Comparison of $\mathcal{B}_{\text{MI}}$ \eqref{B_IT} and $\mathcal{B}_{\text{SEP}}$ \eqref{B_SEP} as a function of the sub-Gaussianity parameter $\alpha \in [0.1, 1]$ for a fixed prior $p=0.4$. 
    In both regimes, our bound $\mathcal{B}_{\text{MI}}$ (\textcolor{plotblue}{blue}) provides a strictly tighter upper bound than $\mathcal{B}_{\text{SEP}}$ (\textcolor{plotorange}{orange}).}
    \label{fig:comparison}
\end{figure}

\begin{enumerate}
    \item \textbf{Our Mutual Information Bound ($\mathcal{B}_{\text{MI}}$):} 
    Derived from Theorem \ref{exp_gen_bound}, this bound relies on the total mutual information between the input and the output system.
Due to the independence of the test and train systems conditioned on $Z$ in this example, the term $I(Z, \Te; W B')$ simplifies, yielding:
    \begin{align}
        \textcolor{plotblue}{\mathcal{B}_{\text{MI}} = \sqrt{2\alpha^2 I(Z; W B')}}.
\label{B_IT}
    \end{align}

    \item \textbf{The Separated Bound from \cite{WDH2025} ($\mathcal{B}_{\text{SEP}}$):} 
    We compare against the bound in \cite[Eq.
(88), Theorem 1]{WDH2025}, which separates the classical and quantum contributions.
In this specific toy example, the first term of their theorem vanishes, reducing the bound to:
    \begin{align}
        \textcolor{plotorange}{\mathcal{B}_{\text{SEP}} = \mathbb{E}_{Z,W} \left[\sqrt{2\alpha^2 D\left(\sigma^{B'}_{(Z,W)} \middle\| \sigma^{B'}_{W}\right)}\right] + \sqrt{2\alpha^2 I(Z;W)}}.
\label{B_SEP}
    \end{align}
\end{enumerate}

As detailed in Figure \ref{fig:comparison} below, for comparison, we plot \eqref{B_IT} and \eqref{B_SEP} for the example mentioned in \cite[Subsection-\RNum{6}]{WDH2025}.
\subsection{Comparision of Theorem \ref{exp_gen_bound} with Theorem $2$ of \cite{WDH2025} }\label{comp:warsi}

In Theorem $2$ of \cite{WDH2025}, the authors obtained an upper-bound on expected generalization bounds of quantum learning algorithms, in terms of R\'enyi divergences.
We contrast our upper-bound on the expected generalization error (Theorem \ref{exp_gen_bound}) with their results,
highlighting the key advantages of our framework.
\begin{itemize}
    \item \textbf{Simplified Sub-Gaussianity Assumptions:}  In \cite[Theorem $2$]{WDH2025}, the authors introduce five distinct point-wise sub-Gaussianity assumptions, as specified in \cite[Assumption $5$ and $6$]{WDH2025}.
These conditions are needed there to establish an upper bound in terms of the R\'enyi divergence.
In contrast, Theorem \ref{exp_gen_bound} replaces Assumptions $5$ and $6$ in \cite{WDH2025} with a single requirement, given in Definition \ref{ass_cq_subg}.
A key difference is that our condition is formulated in expectation with respect to $P_S \times P_W$, rather than being imposed for all worst-case pairs.
\item \textbf{Unified Information Measure:}
The upper bound derived in \cite[Theorem $2$]{WDH2025} involves two terms based on quantum R\'enyi divergence and one term based on classical R\'enyi divergence.
This structure stems from the separate sub-Gaussianity assumptions formulated in \cite[Assumption $5$ and $6$]{WDH2025}.
By contrast, our bound in Theorem \ref{exp_gen_bound} is expressed in terms of a single information-theoretic quantity that simultaneously captures both classical and quantum dependencies.
\end{itemize}

\subsection{Comparison of Theorem \ref{thm:renyi_gen_bound} with Theorem 4 of \cite{WDH2025}}
\label{app:comparison_wdh}

We contrast our concentration bound (Theorem \ref{thm:renyi_gen_bound}) with Theorem 4 of \cite{WDH2025}, highlighting three key advantages of our framework.
\begin{itemize}
    \item \textbf{Simplified Sub-Gaussianity Assumptions:} The result in \cite{WDH2025} necessitates a quantum learning framework with separate pointwise sub-Gaussianity conditions for the quantum posterior states and the classical hypothesis distribution.
We simplify these requirements significantly, relying on only a single  sub-Gaussian condition (Eq. \eqref{ass_cq_subg_loc_eq}) on the individual loss operators.
\item \textbf{Unified Information Measure:} The framework in \cite{WDH2025} employs a "separated" approach that sums classical mutual information and a distinct quantum divergence term.
This decoupling is analytically complex and often results in looser bounds.
In contrast, our bound relies on a single global divergence, $\Tilde{D}_{\gamma} (\sigma^{S\Te WB'}_{\cN}\| \sigma^{S\Te} \otimes \sigma^{WB'}_{\cN})$, which captures classical and quantum dependencies jointly within a unified measure.
\item \textbf{Average-Case vs. Worst-Case:} The quantum term in \cite{WDH2025} is formulated as a "worst-case" bound, typically involving a supremum over inputs or hypotheses (e.g., $\sup_{w} \Tilde{D}_{\gamma}$).
Conversely, our bound is formulated for the average case: it depends on the divergence of the \emph{expected} classical-quantum state, allowing us to directly incorporate the actual data distribution.
\end{itemize}

\section{Proofs of Theorem \ref{lemma_Holevo_stability}, Corollary \ref{corr_Holevo_stability_2} and Corollary \ref{corr_Holevo_stability_1}}\label{proof_lemma_Holevo_stability}

\subsection{Proof of Theorem \ref{lemma_Holevo_stability}}\label{proof_theorem_main}

The proof of Theorem \ref{lemma_Holevo_stability} above relies on the Claims \ref{fact:Dmax_bound} and \ref{lemma_mixture_quant_KL} below.
The proofs of these claims are given in Appendix \ref{proof_fact:Dmax_bound} and \ref{proof_lemma_mixture_quant_KL}.

\begin{claim}
\label{fact:Dmax_bound}
Consider $\rho, \rho', \sigma \in \cD(\cH_A)$.
Then, $D(\rho\|\sigma) \le D(\rho\|\rho') + D_{\max}(\rho'\|\sigma).$
\end{claim}

\begin{claim}\label{lemma_mixture_quant_KL}
Consider $\rho$ and $\sigma$ be two quantum states over Hilbert space $\cH$ such that $\rho \ll \sigma$ and $\sigma$ 
is a finite mixture of probability distributions such that $\sigma=\sum_{b=1}^{m} P(b) \sigma_b$, where $\sum_{b=1}^{m} P(b)=1$, and $\rho \ll \sigma_b$ for all $b \in [m]$.
Then, 
$
D(\rho \| \sigma)
\leq \min _{b\in [m]}\left\{D\left(\rho \| \sigma_b\right)-\ln P(b) \right\}.$
\end{claim}

Since we aim to obtain an upper-bound on $I[S;WB']$, one way to proceed is to use the fact that $I[S;WB']=\min_{\omega^B} D(\sigma^{SB}\|\sigma^S \otimes \omega^B)$ (where $B \equiv WB'$).
Thus,
        \begin{align}
        I[S;WB'] &\leq D(\sigma^{S B} \| \sigma^{S}\otimes\omega^{B})\nn\\
        &= \sum_{s \in \cS} P^{\otimes n}_{Z}(s) D(\cN^{s}(\rho_{s}) \| \omega^B)\label{lemma_Holevo_stability_eq1}.
\end{align}
        
We now choose different values of $\omega^B$ to obtain upper-bounds on $I[S;WB']$ discussed in steps below.

(\textbf{Step ${\bf 1}$}) \label{step1} Consider $\omega^B$ to be a uniform mixture of $\cN^{\bf f}(\rho_{\mathbf{f}}),$ over all the types $\mathbf{f} \in T^{n}_{\abs{\cZ}}$ i.e. $\omega^B := \frac{1}{\abs{T^{n}_{\abs{\cZ}}}}\sum_{\mathbf{f} \in T^{n}_{\abs{\cZ}}}\cN^{\bf f}(\rho_{\mathbf{f}})$, then, using Claim \ref{lemma_mixture_quant_KL} and  the fact that $\abs{T^{n}_{\abs{\cZ}}} \leq (n+1)^{\abs{\cZ}-1}$ (see \cite[Eq. $6.18$]{Group2}), Eq.
\eqref{lemma_Holevo_stability_eq1} can be upper-bounded as follows,
\begin{align}
     &~~~I[S;WB'] \nn\\
     &\leq  \sum_{s \in \cS} P^{\otimes n}_{Z}(s) \min_{\mathbf{f} \in T^{n}_{\abs{\cZ}}} \left\{ D(\cN^s(\rho_{s}) \| \cN^{\bf f}(\rho_{\mathbf{f}})) - \ln \abs{T^{n}_{\abs{\cZ}}}^{-1}\right\}\nn\\
     &\leq(\abs{\cZ}-1) \ln (n+1).\nn
\end{align}

Observe that the above upper-bound on $I[S;WB']$ is independent of the privacy parameters ($\eps,\delta$) of $\cA$.
This happened because, we chose $\omega^B$ to be a uniform mixture of representative quantum states of each type.
This choice implied that $\min_{\mathbf{f} \in T^{n}_{\abs{\cZ}}} D(\cN^s(\rho_{s}) \| \cN^{\bf f}(\rho_{\mathbf{f}})) = 0$.
To get an upper-bound on $I[S;WB']$ in terms of the privacy parameters, we need to choose $\omega^B$ which makes use of the fact that $\cA$ satisfies \eqref{pp4}.
We will accomplish this by using a grid covering for the types of $\cS$. We discuss in the step below.

(\textbf{Step ${\bf 2}$}) In contrast to Step $1$, we will now choose $\omega^{B}$ to be a mixture over a smaller collection of the output states of $\cA$.
This smaller collection is obtained by using a grid covering over the types of $\cS$, which was developed in the proof of Proposition $2$ of \cite{RGBS21}.
We now discuss their grid covering over the types of $\cS$ below.
\quad Observe that any type ${\bf f} \in T^n_{\abs{\cZ}}$ can be 
    thought of as a point inside a $\abs{\cZ}-1$ dimensional grid $[0,n]^{\abs{\cZ}-1}$, which is of size $(n+1)^{\abs{\cZ}-1}$.
This is because, for any ${\bf f} = ({\bf f}_1,\cdots,{\bf f}_{\abs{\cZ}}) \in T^n_{\abs{\cZ}}$, the first $\abs{\cZ}-1$ coordinates decide the last coordinate ${\bf f}_{\abs{\cZ}}$,
     since we have a constraint $\sum_{i=1}^{\abs{\cZ}} {\bf f}_i = n$.
We now split each dimension of the grid $[0,n]^{\abs{\cZ}-1}$ (which is a $[0,n]$ interval) into $t$ equal parts for some 
     \begin{equation}
         t \in \bbN : t \in[1, n].
\label{grid_size_con}
     \end{equation}
     
    \quad  We can think of the grid $[0,n]^{\abs{\cZ}-1}$ as a cover of $t^{\abs{\cZ}-1}$ smaller grids of length $l := \frac{n}{t}$.
Note that each side of the smaller grid has $\floor{l}+1$ points.
Further, if $\floor{l}+1$ is odd, then we choose the central point of the smaller corresponding to the coordinates of the center of the smaller grid.
Thus, for any $s \in \cS$, if we consider its type as ${\bf f}^{(s)}$, 
     then we can find a type ${\bf g}^{(s)} \in T^n_{\abs{\cZ}}$ such that the first $\abs{\cZ}-1$ coordinates of ${\bf f}$ are the coordinates of the center of the smaller grid in which the first $\abs{\cZ}-1$ coordinates of ${\bf f}^{(s)}$ resides.
In each dimension of the bigger grid, the distance between $s$ and the center of the nearest smaller grid ${\bf c}^{s}$ is given as follows,
    \begin{equation*}
        \abs{{\bf f}^{(s)}_{z_{(i)}} - {\bf c}^{(s)}_{z_{(i)}}} \leq \frac{\floor{l}+1}{2} \leq \frac{n}{2t}+\frac{1}{2}, \mbox{~for each } i \in [\abs{\cZ}-1],
    \end{equation*}
    where $z_{(i)}$ is the $i$-th element of the alphabet $\cZ$.
Therefore, if along all dimension $i \in [\abs{\cZ}-1]$, ${\bf f}^{(s)}_{z_{(i)}} - {\bf g}^{(s)}_{z_{(i)}} = -\frac{n}{2t}+\frac{1}{2}$, then the count of last element $z_{(\abs{\cZ})} \in \cZ$ has to compensate for it.
Thus, we have the following,
    \begin{align*}
        \abs{{\bf g}^{(s)}_{z_{(\abs{\cZ})}} - {\bf f}_{z_{(\abs{\cZ})}}} &\leq  (\abs{\cZ}-1)\left(\frac{n}{2t}+\frac{1}{2}\right).
\end{align*} 
    Then, for any $s \in \cS$ the following holds, 
    \begin{align}
          d(s,T_{{\bf g}^{(s)}}) &\leq (\abs{\cZ}-1)\frac{n}{t},\label{lemma_Holevo_stability_eq2}
    \end{align}
    where $d(s,T_{{\bf g}^{(s)}})$ is distance between the types of $s$ and ${\bf g}^{(s)}$ as defined in Section \ref{sec_def_not}.
    
(\textbf{Step ${\bf 3}$}) We now prove Theorem \ref{lemma_Holevo_stability} using the grid covering technique discussed in the proof of \cite[Proposition 2]{RGBS21}.
Fix $\omega^{B} = \sum_{{\bf f} \in T'} \frac{1}{\abs{T'}} \cN^{\bf f}(\rho_{\bf f})$, where $T'$ is the collection of the center points of all the smaller grids obtained in Step 2. Then, using Claim \ref{lemma_mixture_quant_KL} and the fact that $\abs{T'} \leq t^{\abs{\cZ}-1}$, we have,
    {\allowdisplaybreaks\begin{align}
        I[S;WB'] 
        &\leq \sum_{s \in \cS} P^{\otimes n}_{Z}(s) \min_{{\bf f} \in T'} \left\{ D(\cN^{s}(\rho_{s}) \| \cN^{\bf f}(\rho_{\bf f})) + (\abs{\cZ}-1)\ln t\right\}\nn\\
        &\leq \sum_{s \in \cS} P^{\otimes n}_{Z}(s) \left( D(\cN^{s}(\rho_{s}) \| \cN^{{\bf g}^{(s)}}(\rho_{{\bf g}^{(s)}}))+ 
(\abs{\cZ}-1)\ln t\right).\label{lemma_Holevo_stability_eq4}
    \end{align}}

   (\textbf{Step ${\bf 4}$}) We will now analyze the first term in the RHS of \eqref{lemma_Holevo_stability_eq4} by using Claim \ref{fact:Dmax_bound} and \cite[Lemma $6.9$]{Tomamichel_2016}.
Toward this, in \cite[Lemma $6.9$]{Tomamichel_2016}, let $\rho = \cN^s(\rho_s)$ and $\sigma = \cN^{{\bf g}^{(s)}}(\rho_{{\bf g}^{(s)}})$.
Thus, \cite[Lemma $6.9$]{Tomamichel_2016} implies that there exists a quantum state $\cN^s(\rho_{s})'$ in the close vicinity of $\cN^s(\rho_{s})$ such that $D_{\max}(\cN^s(\rho_{s})' \| \cN^{{\bf g}^{(s)}}(\rho_{{\bf g}^{(s)}}) )\leq f(\eps,\delta)$, where $f(\cdot,\cdot)$ is some function.
Thus, using \cite[Lemma $6.9$]{Tomamichel_2016}, Claim \ref{fact:Dmax_bound}, Assumption \eqref{ass1} and the extension of privacy constraints of $\cA$ under $k$-neighboring inputs, we have the following series of inequalities,
    {\allowdisplaybreaks\begin{align}
        D(\cN^s(\rho_{s}) \| \cN^{{\bf g}^{(s)}}(\rho_{{\bf g}^{(s)}}))
        &\leq D(\cN^s(\rho_{s}) \| \cN^s(\rho_{s})')  + \eps'\nn\\
        &\leq \frac{2}{m}E^{2}_{1}\left(\cN^s(\rho_{s}) \| \cN^s(\rho_{s})'\right) + \eps'\nn\\  
       &\leq\eps'+ \frac{2}{m}g_{\frac{n(\abs{\cZ}-1)}{t}}(\eps,\delta) ,\label{lemma_Holevo_stability_eq3}
    \end{align}}
    where $\eps':= \frac{n(\abs{\cZ}-1)\eps}{t} + \ln \frac{1}{1- g_{\frac{n(\abs{\cZ}-1)}{t}}(\eps,\delta)}$, and $g_{\frac{n(\abs{\cZ}-1)}{t}}(\eps,\delta) = \frac{e^{\frac{n(\abs{\cZ}-1)\eps}{t}} - 1}{e^\eps 
- 1}\delta$. Thus,  using Eqs. \eqref{lemma_Holevo_stability_eq4} and \eqref{lemma_Holevo_stability_eq3} we have,
    {\allowdisplaybreaks\begin{align}
        &\hspace{-7pt}I[S;WB'] {\leq}\frac{n(\abs{\cZ}-1)\eps}{t} +  (\abs{\cZ}-1)\ln t + h_{\abs{\cZ}}(\eps,\delta),\label{lemma_Holevo_stability_eq6}
    \end{align}}
   where $h_{\abs{\cZ}}(\eps,\delta):= \ln \frac{1}{1- g_{{n(\abs{\cZ}-1)}}(\eps,\delta)} + \frac{2}{m}g_{n(\abs{\cZ}-1)}(\delta)$ (observe that $h_{\abs{\cZ}}(\eps,0) = 0$) and the last inequality follows from the fact that the grid size $t \geq 1$.
   
(\textbf{Step ${\bf 5}$}) In this step, we optimize the choice over $t$ (grid size) to tighten the upper-bound obtained in Eq.
\eqref{lemma_Holevo_stability_eq6}. Observe that the value of $t$ which minimizes the RHS of Eq.
\eqref{lemma_Holevo_stability_eq6} is,
    \begin{equation}
       t^{\star} = n\eps.\label{optimal_t}
    \end{equation}
    As mentioned in the statement of Theorem \ref{lemma_Holevo_stability}, we have $\eps \in \left[\frac{1}{n}\right., 1]$ and thus \eqref{optimal_t} yields that $1 \leq t^{\star} \leq n$, which satisfies the size constraint of grid mentioned in \eqref{grid_size_con}.
Therefore, substituting $t = t^{\star}$ in \eqref{lemma_Holevo_stability_eq6} yields,    
    \begin{align*}
        I[S;WB'] &\leq (|\mathcal{Z}| - 1) \left(1 + \ln\left({n\eps}\right)\right) + h_{\abs{\cZ}}(\eps,\delta)\nn\\
        &=(|\mathcal{Z}| - 1) \ln\left({n e\eps}\right) + h_{\abs{\cZ}}(\eps,\delta).
\end{align*}
This completes the proof of  Theorem  \ref{lemma_Holevo_stability}.\hfill\QED
\subsection{Proof of Corollary \ref{corr_Holevo_stability_2}}\label{proof_corr1}
For $\eps < \frac{1}{n}$ (as mentioned in Corollary \ref{corr_Holevo_stability_2}), \eqref{optimal_t} yields that $t^{\star} < 1$, and therefore it does not \eqref{grid_size_con}.
Thus, in this case, we set $t = 1$ in \eqref{lemma_Holevo_stability_eq6} to obtain the desired upper-bound.\hfill\QED
\subsection{Proof of Corollary \ref{corr_Holevo_stability_1}}\label{proof_corr2}

    For $\eps > 1$ (as mentioned in Corollary \ref{corr_Holevo_stability_1}), \eqref{optimal_t} yields that $t^{\star} > n$, which does not satisfy the grid size constraint mentioned in \eqref{grid_size_con}.
Therefore in this case we set grid size $t=n$. However, for this choice of grid size, observe that the grid covers all the sequences in $\cS$ and therefore covers all the type-representatives in $\cS$.
This is the same case as Step 1. Therefore, we have, 
    \begin{align}
     I[S;WB']
     &\leq (\abs{\cZ}-1) \ln (n+1)\nn.
\end{align}
This completes the proof of Corollary \ref{corr_Holevo_stability_1}. Further, note that if we substitute $t = n$ in \eqref{lemma_Holevo_stability_eq6}, then it would yield us a weaker bound as compared to the above.
\hfill\QED

\section{Comparison of Upper-bounds on Stability with Prior Work}\label{App:comp_stab}

In this section, we compare the stability upper-bound (Theorem \ref{lemma_Holevo_stability}) derived in this work with existing results.
\subsection{Comparison between Theorem \ref{lemma_Holevo_stability} and \cite[Proposition~10]{Nuradha25}}\label{comp_stab}
In \cite[Proposition~10]{Nuradha25}, the authors derived an upper bound on the Holevo information for quantum $(\varepsilon,\delta)$-LDP quantum channels, as stated in Eq.~(209) of \cite{Nuradha25}.
However, one of the authors of \cite{Nuradha25} later clarified to the authors of the present paper \cite{wilde2025private} that the phrase “for quantum $(\varepsilon,\delta)$-LDP quantum channels” was a typographical error.
The corrected statement is as follows.
If an algorithm $\mathcal{A}$ satisfies $\varepsilon$-QLDP, meaning that
\begin{equation}
\tr[M \mathcal{A}(\rho_x)]
\le e^{\varepsilon}\tr[M \mathcal{A}(\rho_{x'})],
\qquad \forall\, x,x' \in \mathcal{X}, \ \forall\, M: 0 \le M \le I,\label{wilde2}
\end{equation}
then, the following bound holds:
\begin{equation}
I[X;B]_{\sigma} \leq \varepsilon \tanh\left(\frac{\varepsilon}{2}\right) = \varepsilon \left(\frac{e^\varepsilon -1 }{e^\varepsilon +1 }\right),\label{wilde1}
\end{equation}
where $I[X;B]_{\sigma}$ is the Holevo information computed with respect to the state $\sigma := \sum_{x \in \mathcal{X}} P(x) \ketbra{x} \otimes \mathcal{A}(\rho_x)$.
Here,
their learning algorithm $\mathcal {A}$ can be considered as
a map $\cN$ given in \eqref{SHJ}, but the map does not depend on $s$.
Hence, the reference \cite{Nuradha25} also adopts the Untrusted Data Processor scenario similar to the second part of \cite{Caro23}.

In contrast, our main result, Theorem~\ref{lemma_Holevo_stability}, provides an upper bound under a weaker assumption: the algorithm $\mathcal{A}$ satisfies $1$-neighbor $(\varepsilon,\delta)$-DP, i.e., the condition~\eqref{pp4} holds for every $s \overset{1}{\sim} s'$ (see Section \ref{sec_def_not}).
When we set $\delta = 0$ and $\mathcal{X} = \mathcal{S}$, the form of the constraint in \eqref{wilde2} becomes identical to that in \eqref{wilde2}.
However, our result applies this constraint only to $1$-neighboring pairs, whereas their result assumes it for all pairs $s \neq s' \in \mathcal{S}$.

In fact, if we strengthen our assumption to match theirs—namely, require \eqref{pp4} for any distinct $s,s' \in \mathcal{S}$—then part $(i)$ of \cite[Corollary 3]{DWH2025} recovers the same bound as \eqref{wilde1}, thereby aligning our result with the corrected version of \cite[Proposition~10]{Nuradha25}.

The proof of Theorem \ref{exp_gen_bound} formally establishes the connection between algorithmic stability and generalizability by treating the mutual information $I[S;WB']$ as a proxy for stability.
We first demonstrate that the expected generalization error is fundamentally limited by the square root of the information the algorithm leaks about the training data $S$, i.e., we have the following,
\begin{align}
        \overline{\text{\textnormal{gen}}}_{\rho}(\cN) &\leq \sqrt{2\alpha^2 I[S;WB']}.\label{gen_stab_eq}
    \end{align}

Here, $I[S;WB']$ quantifies the dependence of the output hypothesis on the specific training set;
a lower value implies that the algorithm is "stable" and not overfitting to individual data points.
The crucial link to Theorem \ref{lemma_Holevo_stability} is that it provides the explicit upper bound on this stability measure derived solely from the privacy constraints.
Finally, by substituting the bound from Theorem \ref{lemma_Holevo_stability} into \eqref{gen_stab_eq}, we mathematically confirm that the rigorous stability imposed by $(\eps,\delta)$-DP directly suppresses the generalization error, by preventing the algorithm from depending too heavily on any single data point and ensuring that the learned hypothesis performs well on unseen data.
\subsection{Comparison between Theorem \ref{lemma_Holevo_stability} and \cite[Appendix C.7]{Caro23}}\label{comp_stab_car}
The reference \cite{Caro23} studies the local differential privacy of learning algorithms in two settings.
That is, their discussion is composed of two parts, 
the first part starting with ``First'' and the second part starting with ``Next''.
\subsubsection{Assumption in \cite[Appendix C.7]{Caro23}}
Their first part discusses the Holevo information under a certain condition.
However, a careful examination of the proof in \cite[Appendix C.7]{Caro23} reveals that
the argument relies on a stronger assumption than their statement as follows.
In this place, the authors claim to prove the following bound
on the Holevo information.
\begin{equation}
I(\text{test}; \text{hyp})_{\sigma^{\mathcal{A}}_{(s,w)}} \leq 2\varepsilon (1 - e^{-\varepsilon}) \sqrt{2I(\text{test}; \text{train})_{\rho^{\mathcal{A}}_{(s,w)}}},\label{caro1}
\end{equation}
under the assumption that the channel $\Lambda^{\cA}_{s,w} : \cH^{\mbox{train}} \to \cH^{\mbox{hyp}}$ is $\varepsilon$-LDP, i.e.,
\begin{equation}
\tr\left[M \Lambda^{\cA}_{s,w} (\rho^{\mbox{train}}_{1})\right] \leq e^\varepsilon \tr\left[M \Lambda^{\cA}_{s,w} (\rho^{\mbox{train}}_{2})\right],
\label{AL1}
\end{equation}
for all $0 \preceq M \preceq \bbI^{\mbox{hyp}}$ and $\rho^{\mbox{train}}_{1}, \rho^{\mbox{train}}_{2} \in \cD(\cH^{\mbox{train}})$.
However, a closer examination of their proof reveals that the argument implicitly depends on a stronger condition, namely
\begin{equation}
\tr\left[O\big(\bbI^{\mathrm{test}}\otimes \Lambda^{\cA}_{s,w}\big)\big(\rho^{\mathrm{test;train}}_{1}\big)\right]
\leq e^\varepsilon \,
\tr\left[O\big(\bbI^{\mathrm{test}}\otimes \Lambda^{\cA}_{s,w}\big)\big(\rho^{\mathrm{test;train}}_{2}\big)\right],
\label{car_2}
\end{equation}
for all $0 \preceq O \preceq \bbI^{\mathrm{test;hyp}}$ and $\rho^{\mathrm{test;train}}_{1}, \rho^{\mathrm{test;train}}_{2} \in \cD(\cH^{\mathrm{test;train}})$.
In other words, the proof appears to require that $\Lambda^{\cA}_{s,w}$, which is locally $\varepsilon$-LDP on $\cH^{\mathrm{hyp}}$, also preserves $\varepsilon$-LDP globally when extended to the joint space $\cH^{\mathrm{test;hyp}}$.
Crucially, \eqref{AL1} does \emph{not} imply \eqref{car_2}. Indeed, by \cite[Theorem 4]{ZY2017}, the identity channel on $\cH^{\mathrm{test}}$ fails to satisfy differential privacy for any $\varepsilon \geq 0$, so the composition $\bbI^{\mathrm{test}}\otimes \Lambda^{\cA}_{s,w}$ cannot satisfy $\varepsilon$-LDP solely on the basis of \eqref{AL1}.
Therefore, there is a gap in the argument of \cite[Appendix C.7]{Caro23}: the claimed bound \eqref{caro1} does not follow from their stated assumption \eqref{AL1}.
That is, one must assume \eqref{car_2} instead of \eqref{AL1}.
Moreover, the upper bound obtained in \eqref{caro1} involves the term $I(\text{test}; \text{train})_{\rho^{\mathcal{A}}{(s,w)}}$.
To render this bound meaningful, $I(\text{test}; \text{train})_{\rho^{\mathcal{A}}{(s,w)}}$ should also be controlled by some function of the security parameter $\varepsilon$, although our evaluations—such as Theorem \ref{lemma_Holevo_stability}—do satisfy this requirement.
\subsubsection{Security condition in \cite[Appendix C.7]{Caro23}}
Their second part essentially changes their model into 
the  Untrusted Data Processor scenario studied in Section \ref{SSA}
because 
on the page $59$ of \cite{Caro23} the authors mention the following:
    \begin{quote}
        ``Next, we turn our attention to the classical MI term in our generalization bounds.
Here, we assume that the learner $\mathcal{A}$ uses an overall $\varepsilon$-LDP POVM. As the POVM $\left\{|s\rangle\langle s|
\otimes E_s^{\mathcal{A}}(w)\right\}_{s, w}$ is not LDP even if every $\left\{E_s^{\mathcal{A}}(w)\right\}_w$ is, we make the simplifying assumption that the learner uses an $s$-independent $\varepsilon$-LDP POVM $\left\{E^{\mathcal{A}}(w)\right\}_w$.''
    \end{quote}
Even in this scenario, our results still hold, as explained in Section \ref{SSA}.
However, in this scenario, it is reasonable to impose the ITA condition given in Definition \ref{def:ITA} to our learning algorithm,
as discussed in Section \ref{SSA}
while they did not consider such a constraint.

\section{Proof of Corollary \ref{fact_QDP_group}}\label{fact_QDP_group_proof}
Since $s \overset{k}{\sim} s'$, there exists a $k+1$-length sequence $\{s_i\}_{i=0}^{k} \subseteq \cS$ such that $s_0 = s$, $s_k = s'$ and for each $i \in [k]$, $s_{i-1} \overset{1}{\sim} s_i$.
Thus, for any $0 \preceq \Lambda \preceq \bbI$, using Eq \eqref{pp4}, we have
    \begin{align*}
        \tr[\Lambda \cN^{(s)}(\rho_s)] &\leq e^{\eps} \tr[\Lambda \cN^{(s_1)}(\rho_{s_1})] + \delta\\
        &\leq e^{2\eps} \tr[\Lambda \cN^{(s_2)}(\rho_{s_2})] + (e^{\eps}+1)\delta\\
        &\leq e^{3\eps} \tr[\Lambda \cN^{(s_3)}(\rho_{s_3})] + (e^{2\eps}+e^{\eps}+1)\delta\\
        &\vdotswithin{=}\\
        &\leq e^{k\eps} \tr[\Lambda \cN^{(s')}(\rho_{s'})] + (e^{(k-1)\eps}+e^{(k-2)\eps}+\cdots+e^{\eps}+1)\delta\\
        & = e^{k \eps} \tr[\Lambda \cN(\sigma)]+ g_{k}(\delta).
\end{align*}
    This completes the proof of Corollary \ref{fact_QDP_group}.\hfill\QED

\section{Proof of Lemma \ref{lem:ITA_classical}}\label{proof_LL}
Choose a basis $\{\ket{x}\}$ that diagonalizes $\rho_s$.
In this basis, the state can be expressed as
\begin{equation}
\rho_s = \sum_{x} P_{X|s}(x)\,\ketbra{x},
\end{equation}
on system $X$.
Next, define the instrument $\{\mathcal{N}_w'\}_w$ by
\begin{equation}
\mathcal{N}_w'(\ketbra{x}) := \ketbra{x} \otimes \mathcal{N}_w(\ketbra{x}),
\end{equation}
where the output system is $X B'$.
Recall that the original learning algorithm $\{\mathcal{N}_w\}_w$ outputs on system $B'$.
Under this construction, $\{\mathcal{N}_w'\}_w$ is strictly more informative than $\{\mathcal{N}_w\}_w$.

Suppose, for contradiction, that $\{\cN_w\}_w$ is more informative than $\{\cN_{w}'\}_w$.
Then there exist CP-TP maps $\{\Gamma_w\}_w$ such that
\[
\Gamma_{w}(\cN_{w}(\rho_s)) = \cN_{w}'(\rho_s) \quad \text{for all } s\in \cS.
\]
Hence,
\begin{align}
&\tr_{B'W} \Bigg[\sum_w \Gamma_w (\cN_{w}(\rho_s)) \otimes \ketbra{w}\Bigg] \nn\\
=&\tr_{B'W} \Bigg[\sum_w \cN_{w}'\Bigg(\sum_{x}P_{X|s}(x)\ketbra{x}\Bigg)\otimes \ketbra{w}\Bigg]\nn\\
=&\tr_{B'W} \Bigg[\sum_{x}P_{X|s}(x)\sum_w \cN_{w}'(\ketbra{x})\otimes \ketbra{w}\Bigg]\nn\\
=&\tr_{B'W}\Bigg[\sum_{x}P_{X|s}(x)\ketbra{x}\otimes \sum_w \cN_{w}(\ketbra{x})\otimes \ketbra{w}\Bigg]\nn\\
=&\sum_{x}P_{X|s}(x)\ketbra{x} = \rho_s,
\end{align}
which contradicts the assumption that no CP-TP map $\Gamma$ satisfies
\[
\Gamma \left(\sum_{w\in \cW}\cN_{w}(\rho_s)\otimes \ketbra{w}\right) = \rho_s.
\]
Therefore, $\{\cN_w'\}_w$ is strictly more informative than $\{\cN_{w}\}_w$, and thus $\{\cN_{w}\}_w$ is not ITA.\hfill\QED

\section{Proof of Claim \ref{fact:Dmax_bound}}\label{proof_fact:Dmax_bound}
By the definition of the max-relative entropy, we have
\[
    \rho' \le e^{D_{\max}(\rho'\|\sigma)}\, \sigma.
\]
Equivalently,
\[
    \sigma \ge e^{-D_{\max}(\rho'\|\sigma)}\, \rho'.
\]
Since the logarithm is operator monotone, this implies
\[
    \ln \sigma \succeq \ln \rho' - D_{\max}(\rho'\|\sigma)\, \mathbb{I}.
\]
Multiplying both sides by $-\rho$ and taking the trace (which reverses the inequality), we obtain
\[
    -\tr[\rho \ln \sigma] 
    \le -\tr[\rho \ln \rho'] + D_{\max}(\rho'\|\sigma).
\]
Adding $\tr[\rho \ln \rho]$ to both sides gives
\[
    \tr[\rho (\ln \rho - \ln \sigma)]
    \le \tr[\rho (\ln \rho - \ln \rho')] + D_{\max}(\rho'\|\sigma),
\]
which can be written as
\[
    D(\rho\|\sigma) \le D(\rho\|\rho') + D_{\max}(\rho'\|\sigma).
\]
This completes the proof of Claim  \ref{fact:Dmax_bound}.\hfill\QED

\section{Proof of Claim \ref{lemma_mixture_quant_KL}}\label{proof_lemma_mixture_quant_KL}

We begin by invoking the operator monotonicity of the function $\ln(\cdot)$.
Since $\sigma \succeq P(b)\sigma_b$ for every $b$, we obtain
\begin{equation}
\label{eq:logmonotone}
\ln \sigma \succeq \ln P(b)\,\mathbb{I} + \ln \sigma_b.
\end{equation}
Using \eqref{eq:logmonotone}, we immediately have,
\begin{align}
    D(\rho\|\sigma) 
    &= \tr\left[\rho(\ln\rho - \ln\sigma)\right] \nonumber\\
    &\le \tr\left[\rho(\ln\rho - \ln P(b)\,\mathbb{I} - \ln\sigma_b)\right] \nonumber\\
    &= D(\rho\|\sigma_b) - \ln P(b).
\label{eq:simple_bound}
\end{align}
Since \eqref{eq:simple_bound} holds for all $b$, it implies Claim  \ref{lemma_mixture_quant_KL}.
\hfill\QED

\end{document}